\documentclass[aps,floatfix,amsmath,notitlepage,longbibliography]{revtex4-1}

\usepackage{amssymb}
\usepackage{graphicx}
\usepackage{graphics}
\usepackage{amsmath}
\usepackage{amsthm}
\usepackage{color}
\usepackage{dsfont}
\usepackage{mathrsfs}

\usepackage{mathtools}
\usepackage{hyperref}

\DeclarePairedDelimiter{\ceil}{\lceil}{\rceil}

\def\free{\text{free}}
\def\>{\rangle}
\def\<{\langle}

\def\id{\mathsf{id}}
\def\mB{\mathcal{B}}
\def\mE{\mathcal{E}}

\def\mF{\mathcal{F}}
\def\mN{\mathcal{N}}

\def\mL{\mathcal{L}}
\def\mP{\mathcal{P}}

\def\mD{\mathcal{D}}
\def\mT{\mathcal{T}}

\def\mQ{\mathcal{Q}}

\renewcommand{\qedsymbol}{\nobreak \ifvmode \relax \else
	\ifdim \lastskip<1.5em \hskip-\lastskip \hskip1.5em plus0em
	minus0.5em \fi \nobreak \vrule height0.75em width0.5em
	depth0.25em\fi}

\renewcommand{\geq}{\geqslant}
\renewcommand{\leq}{\leqslant}

\def\mb{\mathfrak{B}}
\def\mc{\mathfrak{C}}
\def\md{\mathfrak{D}}
\def\ms{\mathfrak{S}}

\def\mg{\mathfrak{G}}

\def\ml{\mathfrak{L}}

\newtheorem{theorem}{Theorem}
\newtheorem*{theorem*}{Theorem}

\newtheorem{lemma}{Lemma}
\newtheorem*{lemma*}{Lemma}

\newtheorem*{definition}{Definition}
\newtheorem*{definition*}{Definition}

\theoremstyle{remark}
\newtheorem*{remark}{Remark}

\theoremstyle{remark}
\newtheorem*{remark1}{Remark 1}

\theoremstyle{remark}
\newtheorem*{remark2}{Remark 2}

\theoremstyle{remark}

\theoremstyle{remark}

\theoremstyle{definition}

\newcommand{\bea}{\begin{eqnarray}}
\newcommand{\eea}{\end{eqnarray}}
\newcommand{\be}{\begin{equation}}
\newcommand{\ee}{\end{equation}}
\newcommand{\ba}{\begin{equation}\begin{aligned}}
\newcommand{\ea}{\end{aligned}\end{equation}}

\newtheorem{example}{Example}

\def\be{\begin{equation}}
\def\ee{\end{equation}}

\newcommand{\J}{\mathbf{J}}

\newcommand{\cptp}{{\rm CPTP}}
\newcommand{\cp}{{\rm CP}}
\newcommand{\pos}{\rm Pos}

\newcommand{\mI}{\mathcal{I}}

\newcommand{\mR}{\mathcal{R}}
\newcommand{\mM}{\mathcal{M}}

\newcommand{\lr}{\rangle\langle}

\newcommand{\tr}{{\rm Tr}}

\newcommand{\ve}[1]{{\left\vert\kern-0.25ex\left\vert\kern-0.25ex\left\vert #1 
    \right\vert\kern-0.25ex\right\vert\kern-0.25ex\right\vert}}

\newcommand{\mbb}[1]{\mathbb{#1}}

\newcommand{\bra}[1]{\langle #1|}
\newcommand{\ket}[1]{|#1\rangle}

\newcommand{\op}[2]{| #1\rangle\langle #2 |}

\newcommand{\eqdef}{\coloneqq}

\newcommand{\1}{\mathds{1}}

\def\tA{\tilde{A}}
\def\tB{\tilde{B}}

\def\tR{\tilde{R}}

\def\mf{\mathfrak{F}}
\def\mk{\mathfrak{K}}

\def\herm{{\rm Herm}}
\def\free{{\rm FREE}}
\def\disc{{\rm DISC}}
\def\misc{{\rm MISC}}
\def\isc{{\rm ISC}}
\def\sisc{{\rm SISC}}
\def\exact{{\rm exact}}
\def\mio{\rm MIO}

\begin{document}
	
	%\SetWatermarkText{NOT FOR DISTRIBUTION}
	%\SetWatermarkAngle{60}
	%\SetWatermarkScale{0.5}
	
	\title{Dynamical Resource Theory of Quantum Coherence}

\author{Gaurav Saxena}\email{gaurav.saxena1@ucalgary.ca}	
\affiliation{
Department of Physics and Astronomy, Institute for Quantum Science and Technology,
University of Calgary, AB, Canada T2N 1N4}

\author{Eric Chitambar}\email{echitamb@illinois.edu}
\affiliation{Department of Electrical and Computer Engineering, Coordinated Science Laboratory, University of Illinois at Urbana-Champaign,  Urbana,  IL 61801}

\author{Gilad Gour}\email{gour@ucalgary.ca}
\affiliation{
Department of Mathematics and Statistics, Institute for Quantum Science and Technology,
University of Calgary, AB, Canada T2N 1N4}

	\date{\today}
	
\begin{abstract}
	Decoherence is all around us. Every quantum system that interacts with the environment is doomed to decohere.
	The preservation of quantum coherence is one of the major challenges faced in quantum technologies, but its use as a resource is very promising and can lead to various operational advantages,
	for example in quantum algorithms.
	Hence, much work has been devoted in recent years to quantify the coherence present in a system.
	In the present paper, we formulate the quantum resource theory of dynamical coherence.  The underlying physical principle we follow is that the free dynamical objects are those that cannot preserve or distribute coherence.  This leads us to identify classical channels as the free elements in this theory.  Consequently, even the quantum identity channel is not free as all physical systems undergo decoherence and hence, the preservation of coherence should be considered a resource.  In our work, we introduce four different types of free superchannels (analogous to MIO, DIO, IO, and SIO) and discuss in detail  two of them, namely, dephasing-covariant incoherent superchannels (DISC), maximally incoherent superchannels (MISC).  The latter consists of all superchannels that do not generate non-classical channels from classical ones.
	We quantify dynamical coherence using channel-divergence-based monotones for $\misc$ and $\disc$. We show that some of these monotones have operational interpretations as the exact, the approximate, and the liberal coherence cost of a quantum channel. Moreover, we prove that the liberal asymptotic cost of a channel is equal to a new type of regularized relative entropy.
	Finally, we show that the conversion distance between two channels under $\misc$ and $\disc$ can be computed using a semi-definite program (SDP).

\end{abstract}

	\maketitle
%{
%  \hypersetup{linkcolor=black}
%  \tableofcontents
%}	

\section{Introduction}

Decoherence is everywhere.
All physical systems undergo decoherence.
It is an irreversible process, and it can be viewed as the reduction of a general quantum state to an incoherent mixed state due to coupling with the environment \cite{Z2003, S2005, L2003}.
Mathematically, decoherence is represented as the vanishing of the off-diagonal terms of a density matrix.  It is because of decoherence that we do not observe quantum mechanical behaviour in everyday macroscopic objects, and in the context of quantum information, it can be viewed as the loss of information from a system into the environment \cite{B2001}.

During the last two decades, interest in quantum information science has shifted towards using quantum mechanical phenomena (like entanglement, nonlocality, etc.) as resources to achieve something that is otherwise not possible through classical physics (eg., quantum teleportation) \cite{HHH+2009, PV2007, BRS2007, GHR+2016, GMN+2015, BCP+2014, WPG+2012, RHP2014, ABC2016, MBC+2012}.
Quantum resource theories (QRTs) use this resource-theoretic approach to exploit the operational advantage of such phenomena and to assess their resource character systematically \cite{CG2019}.
%In principle, every property of quantum mechanics that is not present in classical physics could lead to an operational advantage~\cite{SAP2017}.  
The preservation of quantum coherence is crucial for building quantum information devices, since the loss of quantum superposition due to decoherence negates any non-classical effect in a quantum system \cite{Z2003, Z1991, Z2007}.
Hence from a technological perspective, there is increasing interest in developing a resource theory of coherence \cite{CG2019}.
%This motivates interest in quantification of non-classicality which can be studied in an organized way through the QRT of coherence.
In addition, the resource-theoretic study of quantum coherence might provide new insights towards distinguishing classical and quantum physics in a quantitative manner.
Some other examples of quantum resource theories include the QRT of entanglement, thermodynamics, magic states, Bell non-locality, etc.

Most quantum resource theories are governed by the constraints arising from physical or practical settings. These constraints then lead to the operations that can be freely performed.
For instance, in the static resource theory of quantum entanglement, for any two spatially separated but possibly entangled systems, the spatial separation puts the restriction that only local operations along with classical communications (LOCC) can be performed~\cite{PV2007, HHH+2009, BBP+1996, BDS+1996, LP2001}.
Given this restriction, only separable states can be generated using LOCC, which makes them the free states of the theory.
But unlike entanglement and other constraint-based QRTs, coherence is a state-based QRT.
This means there is no natural set of physical restrictions or practical constraints that strongly motivate a particular set of free operations.
 Instead, the free states are the physically-motivated objects, and the free operations are not unique, only being required to satisfy the basic golden rule of a QRT, i.e., the free operations should be completely resource non-generating (CRNG) \cite{CG2019}. 

In the static resource theory of quantum coherence there is a fixed basis, the so-called classical or incoherent basis, and 
the set of density matrices that are diagonal in this basis form the free states of the theory.
Such states are called incoherent states.
The free operations are then some set of quantum channels that map the set of incoherent states to itself.
The most well-studied classes of free operations are the maximally incoherent operations (MIO), the incoherent operations (IO), the dephasing-covariant incoherent operations (DIO), and sthe trictly  incoherent operations (SIO)~\cite{A2006, BCP2014, CG2016a, CG2016b, CG2016c, MS2016}.  However, all of these operations cost coherence to be implemented even though they cannot generate coherence, i.e. they are CRNG, which means they do not admit a free dilation \cite{CG2016a, CG2016b, CG2016c, MS2016}.
 Therefore, they cannot truly be considered free.
However, one can still use these operations to study static coherence since they cannot increase the coherence in a state, and hence they allow for a comparison of the coherence in two different states based on state convertibility.
 Consequently, a large amount of work has been devoted to developing the theory of static coherence under these operations \cite{SAP2017}.%\cite{BCP2014, CG2016a, CG2016b, CG2016c,DBG2015, DBG2017, ZMC+2017, SRB+2017}. 

Taking this into consideration, we argue here that, contrary to some other works on dynamical coherence \cite{DKW+2018, BGM+2017, CH2016, LY2019, TEZ+2019}, one should look beyond MIO, DIO, IO, and SIO to understand coherence in a dynamical manner since these operations have the ability to preserve and distribute coherence.  Indeed, even the identity channel should be considered as a resource since all physical systems undergo decoherence, and the preservation of coherence should thus be considered a resource. 
Note also that, in quantum computing, diagonal unitaries such as the $T$-gate, are often considered as resources.
Therefore, even some channels in the set of physically incoherent operations (PIO) introduced in~\cite{CG2016a} will be considered resourceful here, as we do not assume that diagonal unitaries are free.

Our approach is therefore to apply the same philosophy of static coherence to dynamical coherence.
This can be done as follows. We take the set of all classical channels to be free in analogy to the static case where all free states are classical.  Here, a channel $\mN_A$ is called classical if and only if
\begin{equation}
\label{Eq:classical-channel-defn}
\mN_{A} = \mD_{A_1} \circ \mN_{A} \circ \mD_{A_0},
\end{equation}
where $\mD_{A_0}\text{ and } \mD_{A_1}$ are dephasing channels for systems $A_0$ and $A_1$ in some fixed basis, respectively; i.e. $\mD_{A_0}(\rho)=\sum_{i=1}^{d_{A_0}}\op{i}{i}\rho\op{i}{i}_{A_0}$, and similarly for $\mD_{A_1}$.  We will denote the set of classical channels that take system $A_0$ to $A_1$ by $\mc(A_0\to A_1)$,
\be
\mN_A \in \mc(A_0 \to A_1) \iff \mN_{A} = \mD_{A_1} \circ \mN_{A} \circ \mD_{A_0}.
\ee
In particular, the identity channel $\id_{A_0\to A_1}$ (here, $A_0$ and $A_1$ correspond to the same system in two different temporal or spatial locations and so, $|A_0| = |A_1|$) is not classical as it does not satisfy the above condition.
Here, the identity channel corresponds to the preservation of coherence for a certain given amount of time.
  Also note the similarity here between the dynamical free objects defined in Eq. \eqref{Eq:classical-channel-defn} and the static free objects in coherence theory.  On the level of states, a density operator $\rho$ is incoherent with respect to the fixed basis if 
\begin{equation}
\rho=\mD_{A_1}(\rho).
\end{equation} 
In fact, this can be seen as a special case of Eq. \eqref{Eq:DISC-distill-cons1} when system $A_0$ is one-dimensional.  Therefore, we identify non-classical channels as those possessing dynamical coherence.

 Like MIO in the QRT of static coherence, we define maximally incoherent superchannels (MISC) to be the set of all superchannnels that do not generate non-classical channels from the classical ones.
  Similar to MIO in the static case, MISC cannot be implemented without coherence-generating channels.
  For example, if we take the pre-processing channel to be any detection-incoherent channel (as defined in \cite{TEZ+2019}) and the post-processing channel to be any maximally incoherent channel, then we get a superchannel which belongs to MISC but its pre- and post-processing channels are non-classical.
 Nonetheless, much like the argument in static coherence, since we are interested in quantifying the coherence of a channel (as opposed to the coherence of a superchannel), we can use such superchannels as they cannot generate coherence at the channel level.
  That means, the superchannel might be composed of non-classical pre- and post-processing channels, but (even if it is tensored with the identity superchannel) it does not output a non-classical channel whenever the input is a classical channel (i.e. it is CRNG).

In our work, we formulate the QRT of dynamical coherence.
We define four different sets of free superchannels: maximally incoherent superchannels (MISC), dephasing-covariant incoherent superchannels (DISC), incoherent superchannels (ISC), and strictly incoherent superchannel (SISC), which are the analog of MIO, DIO, IO and SIO, in the static case. 
  We focus specifically on the QRTs of MISC and DISC.
  Similar to how MIO is defined with respect to the dephasing channel, we define MISC with respect to dephasing superchannel, $\Delta$ (whose pre- and post-processing channels are dephasing channels) in the following way
 \be
\Theta\in\misc(A\to B)\iff \Delta_B\circ\Theta_{A\to B}\circ\Delta_A=\Theta_{A\to B}\circ\Delta_A\, .
 \ee
where $\misc(A \to B)$ means that the superchannel $\Theta$ converts a quantum channel that takes system $A_0$ to $A_1$ to another quantum channel that takes system $B_0$ to $B_1$.
  Its illustration is given in figure \ref{MISC}.
  \begin{figure}[h]
   \includegraphics[width=0.95\textwidth]{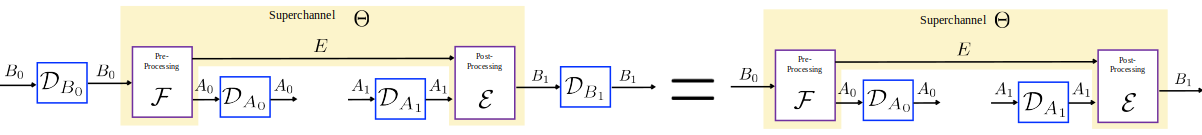}
  \caption{\linespread{1}\selectfont{\small MISC}} 
  \label{MISC}
\end{figure}\newline
DISC is defined analogously to how DIO is defined in static coherence, i.e., 
\be 
\Theta\in\disc(A\to B)\iff\Delta_B\circ\Theta_{A\to B}=\Theta_{A\to B}\circ\Delta_A
\ee
and its illustration is given in figure \ref{DISC}.
In our work, we provide simple characterization of $\misc$ and $\disc$.
 \begin{figure}[h]
   \includegraphics[width=0.95\textwidth]{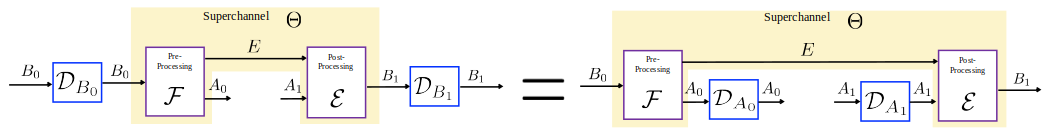}
  \caption{\linespread{1}\selectfont{\small DISC}} 
  \label{DISC}
\end{figure}

 We also quantify dynamical coherence using techniques from QRT of quantum processes \cite{GW2019, LW2019, LY2019} and study the interconversion of channels (i.e., simulation of one channel with another) under $\misc$ and $\disc$.
For the quantification of dynamical coherence, we list here a few key definitions. First, we define the relative entropy of dynamical coherence  under MISC to be (for any quantum channel $\mN_A \in \cptp(A_0 \to A_1)$) 
 \ba
C\left(\mN_A\right)&\eqdef\min_{\mM\in\mc(A_0\to A_1)}D\left(\mN_A\big\|\mM_A\right)\\
&\eqdef\min_{\mM\in\mc(A_0\to A_1)}\max_{\phi \in \md(R_0A_0)}D\left(\mN_{A_0 \to A_1}\left(\phi_{R_0A_0}\right)\big\|\mM_{A_0 \to A_1}\left(\phi_{R_0 A_0}\right)\right)
\ea 
 where $\mc(A_0 \to A_1)$ denotes the set of all classical channels, $\md(R_0A_0)$ denotes the set of density matrices on system $R_0A_0$, and $D(\rho \| \sigma) = \tr[\rho\log \rho - \rho\log \sigma]$ is the relative entropy.
 This monotone is faithful, i.e., zero iff $\mN_A \in \mc(A_0 \to A_1)$, and does not increase under MISC.
 For $\disc$, we define the relative entropy of dynamical coherence to be the function $D_{\Delta}$, given by
\be
D_\Delta(\mN_A)\eqdef D\left(\mN_A\big\|\Delta_A\left[\mN_A\right]\right)\;.
\ee
We show that it is a faithful monotone under DISC.

Similarly, the log-robustness of dynamical coherence is defined as
\be
  LR_\mc(\mN_{A}) \eqdef\min_{\mE\in\mc(A_0\to A_1)}
                 D_{\max}\big(\mN_{A}\|\mE_{A}\big)
\ee
and the dephasing log-robustness of dynamical coherence as 
\be
LR_\Delta(\mN_A)\eqdef D_{\max}\big(\mN_A\big\|\Delta_A[\mN_A]\big)\quad\forall\;\mN\in\cptp(A_0\to A_1)\;.
\ee
We prove that both these quantities are additive under tensor product and have  operational interpretations as the exact dynamical coherence costs in the $\misc$ and $\disc$ cases, respectively.

We then compute the liberal asymptotic cost of dynamical coherence (which is the dynamical coherence cost of a channel when the smoothing is ``liberal"~\cite{GW2019})  under $\misc$, and show that it is equal to a variant of the regularized relative entropy  given by
\be
D^{(\infty)}_{\mc}(\mN_A) \eqdef\lim_{n\to\infty}\frac{1}{n}\sup_{\varphi\in\md(RA_0)}\min_{\mE\in\mc(A^n_0\to A^n_1)}
D\left(\mN^{\otimes n}_{A_0\to A_1}\left(\varphi_{RA_0}^{\otimes n}\right)
\big\|\mE_{A^n_0\to A^n_1}\left(\varphi_{RA_0}^{\otimes n}\right)\right)
\ee

Moreover, we define the interconversion distance, $d_{\mf}(\mN_{A} \to \mM_{B})\eqdef\min_{\Theta\in\mf(A\to B)}\frac{1}{2}\left\|\Theta_{A\to B}\left[\mN_A\right]-\mM_B\right\|_\diamond$ between two quantum channels, $\mN_A \in \cptp(A_0 \to A_1) \text{ and } \mM_B \in \cptp(B_0 \to B_1)$ and show that if $\mf= \misc  \text{ or } \disc$, then $d_{\mf}(\mN_{A} \to \mM_{B})$ can be computed using a semi-definite program (SDP). Lastly, we formulate the one-shot distillable dynamical coherence and compute its value for a few specific channels, including the identity channel.

\section{Preliminaries}

\subsection{Notations}

In this article, we will denote all the dynamical systems and their corresponding Hilbert spaces by $A, B, C,$ etc, and all the static systems and their corresponding Hilbert spaces by $A_1, B_1, C_1$, etc.
In this setting, the notation for a dynamical system, say $A$, indicates a pair of systems such that $ A = (A_0, A_1) = (A_0 \to A_1)$ where $A_0$ and $A_1$ represent the input and output systems, respectively.
The choice of notation for the static systems is because all the states can be viewed as channels with trivial input.
For a composite system, the notation like $A_0B_0$ will be used to mean $A_0 \otimes B_0$.
To represent the dimension of a system, two vertical lines will be used. For example, the dimension of system $A_0$ is $|A_0|$.
A replica of the same system would be represented by using a tilde symbol.
For instance, system $\tA_0$ is a replica of system $A_0$, and system $\tA_1 \tB_1$ is a replica of system $A_1B_1$ i.e., $|\tA_0| = |A_0|$ and $|\tA_1 \tB_1| = |A_1B_1|$.

The set of bounded operators, Hermitian operators, positive operators and density matrices on system $A_0$ would be denoted by $\mb(A_0)$, $\herm(A_0)$, $\pos(A_0)$, and $\md(A_0)$, respectively.
Note that $\md(A_0)\subset \pos(A_0) \subset\herm(A_0) \subset \mb(A_0)$.
Density matrices would be represented by lowercase Greek letters $\rho$, $\sigma$, $\tau$, etc.
We will denote the maximally coherent state (or the plus state) by $\phi^+_{B_1}$ for a system $B_1$ and the unnormalized maximally entangled states by $\phi^{+}_{A_1 B_1}$ for a bipartite system $A_1 B_1$ (note the subscripts in both).
The maximally mixed state for a system $B_1$ will be denoted by $u_{B_1}$.
The set of all linear maps from $\mb(A_0)$ to $\mb(A_1)$ would be denoted by $\ml(A_0 \to  A_1)$, the set of all completely positive maps from $\mb(A_0) \to \mb(A_1)$ would be denoted by $\cp(A_0 \to A_1)$ and the set of quantum channels would be denoted by $\cptp(A_0 \to A_1)$ with $\cptp(A_0 \to A_1) \subset \cp(A_0 \to A_1) \subset \ml(A_0 \to A_1)$.
Throughout this article, we would use calligraphic letters like $\mE, \mF, \mM, \mN,$ etc, to represent quantum channels. 
For simplicity, we will denote a quantum channel with a subscript $A$, like $\mE_A$, to denote an element of $\cptp(A_0 \to A_1)$.
The identity map in $\ml(A_0 \to A_0)$  will be denoted by $\id_{A_0}$.

The notation $\ml(A \to B)$ will be used to denote the set of all maps from $\ml(A_0 \to A_1)$ to $\ml(B_0 \to B_1)$.
Similarly, the set of all maps from Herm$(A_0 \to A_1)$ to Herm$(B_0 \to B_1)$ would be denoted by Herm$(A \to B) \subset \ml(A \to B)$. Identity superchannel in $\ml(A \to A)$ would be denoted by $\1_A$.
All linear maps in $\ml(A \to B)$ and $\herm(A \to B)$ are known as supermaps and the set of supermaps that map quantum channels to quantum channels (even when tensored with the identity supermap) are called superchannels~\cite{CDP2008,G2019}.
We will use capital Greek letters like $\Theta, \Sigma, \Omega$, etc, to denote supermaps.
Square brackets will be used to denote the action of supermaps on linear maps.
For instance, $\Theta_{A \to B}[\mE_A]$ is a linear map in $\ml(B_0 \to B_1)$ obtained by the action of a supermap $\Theta \in \ml(A \to B)$ on a map $\mE \in \ml(A_0 \to A_1)$.
More detailed description of the supermaps and superchannels is provided in the next subsection.

\subsection{Supermaps and Superchannels}

The space $\ml(A_0 \to A_1)$ is equipped with the following inner product
\be \label{inner_prod_maps}
\left\langle \mN_A, \mM_A \right\rangle \eqdef \sum_{i,j} \left\langle \mN_A\left(\ket{i}\bra{j}_{A_0} \right), \mM_A\left( \ket{i}\bra{j}_{A_0} \right)\right\rangle_{HS}
\ee
where $\left\langle X,Y \right\rangle_{HS} \eqdef \tr[X^{*}Y]$ is the Hilbert-Schmidt inner product between the matrices $X,Y \in  \mb(A_1)$.
The above inner product is independent of the choice of the orthonormal basis $\left\{ \ket{i}\bra{j} \right\} \in \mb(A_0)$, and can be expressed in terms of Choi matrices.
The Choi matrix of a channel $\mN_A$ is given by
\be
J^{\mN}_{A_0A_1} = \mN_{\tA_0 \to A_1}\left( \phi^+_{A_0\tA_0} \right)
\ee
where $\phi^+_{A_0\tA_0} \equiv \ket{\phi^+}\bra{\phi^+}_{A_0\tA_0}$ is an unnormalized maximally entangled state where $\ket{\phi^+}_{A_0\tA_0} \equiv \sum_{i}^{|A_0|}\ket{i}_{A_0.}\ket{i}_{\tA_0}$.
With this notation, the inner product of two channels $\mN_A$ and $\mM_A$ can be expressed as
\be
\left\langle \mN_A, \mM_A \right\rangle = \left\langle J^{\mN}_A, J^{\mM}_A \right\rangle = \tr\left[\left(J^{\mN}_A\right)^{*} J^{\mM}_A  \right]
\ee

The canonical orthonormal basis of $\ml(A)$ (relative to the above inner product) is given by $\left\{\mE^{ijkl}_A \right\}$ where
\be
\mE^{ijkl}_A\left( \rho_{A_0} \right) = \bra{i}\rho_{A_0}\ket{j}\; \ket{k}\bra{l}_{A_1} \quad \forall \; \rho_{A_0} \in \mb(A_0)
\ee

The space $\ml(A \to B)$ (where $A = (A_0, A_1) \text{ and } B = (B_0, B_1)$) is equipped with the following inner product
\be
\left\langle \Theta_{A\to B}, \Omega_{A \to B} \right\rangle \eqdef \sum_{i,j,k,l} \left\langle \Theta_{A\to B}\left[\mE^{ijkl}_A \right], \Omega_{A\to B}\left[ \mE^{ijkl}_A \right] \right\rangle
\ee
where $\Theta_{A\to B}, \Omega_{A\to B} \in \ml(A\to B)$ and the inner product on the right-hand side is the inner product between maps as defined in \eqref{inner_prod_maps}.
Similar to how we can express the inner product of two maps by the inner product of their Choi matrices, we can define the inner product of two supermaps as the inner prouct of their Choi matrices as well.
The Choi matrix of a supermap $\Theta_{A\to B}$ is defined as \cite{G2019}
\be
\J^{\Theta}_{AB} = \sum_{i,j,k,l} J^{\mE^{ijkl}}_A \otimes J^{\Theta[\mE^{ijkl}]}_B
\ee
where $ J^{\mE^{ijkl}}_A \text{ and } J^{\Theta[\mE^{ijkl}]}_B$ are the Choi matrices of $\mE^{ijkl}_A$ and $\Theta_{A\to B}[\mE^{ijkl}_A]$, respectively.
With this notation, the inner product between two supermaps $\Theta_{A\to B} \text{ and } \Omega_{A\to B}$ can be expressed as 
\be
\left\langle \Theta_{A\to B}, \Omega_{A \to B} \right\rangle = \left\langle \J^{\Theta}_{AB}, \J^{\Omega}_{AB}\right\rangle_{HS}  = \tr\left[ \left( \J^{\Theta}_{AB} \right)^{*} \J^{\Omega}_{AB} \right]
\ee

We now give three alternate expressions of the Choi matrix of the supermap $\Theta \in \ml(A\to B)$ \cite{G2019}.
First, from its defintion, the Choi matrix of a supermap uses the $\cp$ map analog of entangled states which we represent as $\mP^+_{A\tA}$ and is given by
\be
\mP^+_{A\tA} = \sum_{i,j,k,l}\mE^{ijkl}_{A_0 \to A_1} \otimes \mE^{ijkl}_{\tA_0 \to \tA_1}
\ee
Similar to the properties of the maximally entangled state, the channel $\mP^{+}_{A\tA}$ satifies the following relation for any $\Theta \in \ml(A \to B)$
\be
\Theta_{\tA \to B}[\mP^{+}_{A\tA}] = \Theta^{T}_{\tB \to A}[\mP^{+}_{A\tA}]
\ee
where $\Theta^T \in \ml(B \to A)$ is the transpose of the supermap $\Theta$ which is defined by its components
\be
\left\langle \mE^{ijkl}_A, \Theta^{T}\left[ \mE^{i'j'k'l'}_B \right] \right\rangle = \left\langle \mE^{i'j'k'l'}_B , \Theta\left[ \mE^{ijkl}_A \right]\right\rangle \quad \forall\; i, j, k, l, i', j', k', l'
\ee
where $\left\{  \mE^{ijkl}_A\right\} $ and $\left\{  \mE^{i'j'k'l'}_B \right\}$ are the canonical orthonormal basis of $\ml(A)$ and $\ml(B)$, respectively.
Then, the Choi matrix of a superchannel $\Theta \in \ml(A\to B)$ can be expressed as
\be
\J^{\Theta}_{AB} = \Theta\left[\mP^{+}_{A\tA}\right]\left(\phi^{+}_{A_0\tA_0}\otimes \phi^{+}_{B_0\tB_0}\right)
\ee

The second way of defining the Choi matrix of a supermap is by its action on the Choi matrices of channels.
Lets consider a linear map $\Theta$ such that for $\mM_A \in \ml(A)$ and $\mN_B \in \ml(B)$,  $\mN_B = \Theta_{A\to B}[\mM_A]$.
Then the Choi matrices of $\mM_A$ and $\mN_A$ are related via
\be
J^{\mN}_{B} = \tr_{A}\left[ \J^{\Theta}_{AB}\left(\left(J^{\mM}_A\right)^{T} \otimes I_B  \right) \right]
\ee
That is, $\J^{\Theta}_{AB}$ can be interpreted as the Choi matrix of a linear map (say $\mR^{\Theta}_{A\to B}$) that converts $J^{\mM}_A$ to $J^{\mN}_B$.

For the last representation of the Choi matrix of a supermap, we can view it as a linear map $\mQ^{\Theta} : \mb(A_1B_0) \to \mb(A_0B_1)$ which is defined by the map satisfying
\be
\J^{\Theta}_{AB} \eqdef \mQ^{\Theta}_{\tA_1\tB_0 \to A_0B_1}(\phi^{+}_{A_1\tA_1}\otimes \phi^{+}_{B_0\tB_0}) 
\ee
We will see that the three representations play a useful role in our study of dynamical resource theory of coherence.

The dual of a linear map $\Theta \in \ml(A\to B)$ is a linear map $\Theta^* \in \ml(B\to A)$ with the property for all $\mM_A \in \ml(A)$ and for all $\mN_A \in \ml(B)$
\be
\left\langle \mN_B, \Theta\left[ \mM_A\right] \right\rangle = \left\langle \Theta^*\left[\mN_B \right], \mM_A \right\rangle 
\ee

Now let us define a superchannel.
A superchannel is a supermap $\Theta \in \ml(A\to B)$ that takes quantum channels to quantum channels even when tensored with identity supermap \cite{G2019, CDP2008, P2017, CDP2009, CDP+2013, BP2019, BGS+2019}.
The following are equivalent \cite{CDP2008, G2019}:
\begin{enumerate}
\item $\Theta$ is a superchannel
\item The Choi matrix $\J^{\Theta}_{AB}$ with marginals
\be
J^{\Theta}_{A_1 B_0} = I^{A_1B_0} \quad ; \quad J^{\Theta}_{AB_0} = J^{\Theta}_{A_0 B_0} \otimes u_{A_1}
\ee
where $u_{A_1} = \frac{I_{A_1}}{|A_1|}$ is the maximally mixed state for system $A_1$.
\item The map $\mR^{\Theta}_{A\to B}$ is $\cp$, and there exists a unital $\cp$ map $\mR^{\Theta}_{A_0 \to B_0}$ such that the map $\mR^{\Theta}_{A\to B_0} \equiv \tr_{B_1} \circ \mR^{\Theta}_{A\to B}$ satisfies
\be
\mR^{\Theta}_{A\to B_0} = \mR^{\Theta}_{A_0\to B_0} \circ \tr_{A_1}
\ee
\item There exists a Hilbert space $E$, with $|E| \leq |A_0B_0|$, and two $\cptp$ maps $\mF\in \cptp(B_0 \to A_0E)$ and $\mE \in \cptp(A_1E \to B_1)$ such that for all $\mN_A \in \ml(A_0 \to A_1)$
\be 
\Theta[\mN_A] = \mE_{A_1E \to B_1}\circ \mN_{A_0 \to A_1} \circ \mF_{B_0\to A_0E}
\ee
(see figure \ref{superchannel})
\end{enumerate}
\begin{figure}[h]
   \includegraphics[width=0.65\textwidth]{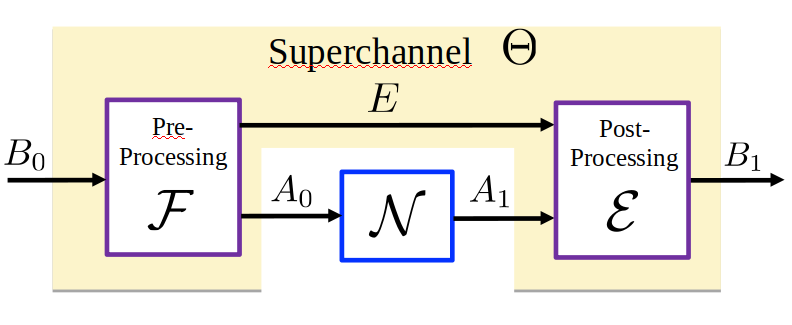}
  \caption{\linespread{1}\selectfont{\small Realization of a superchannel in terms of pre- and post-processing channels}} 
  \label{superchannel}
\end{figure}

\subsection{QRT of static coherence}

Coherence of a state is a basis-dependent concept.
Hence, a basis is fixed first in the resource theory of static coherence.
The density matrices that are diagonal in this basis form the free states of the theory.
These states are also called incoherent states.
Let us denote this set by $\mI_{A_1} \subset \mb(A_1)$ for any system $A_1$.
Hence, all the incoherent density operators $\rho_{A_1} \in \mI_{A_1}$ have the following form
\be
\rho_{A_1} = \sum_{i=0}^{|A_1| - 1}p_i\ket{i}\bra{i}_{A_1}
\ee
with probabillities $p_i$ and obey
\be
\mD_{A_1}(\rho_{A_1}) = \rho_{A_1}
\ee
where $\mD_{A_1}$ is the dephasing channel for the system $A_1$ and is defined as 
\be
\mD_{A_1}(\sigma_{A_1}) = \sum_{i=0}^{|A_1|-1}\ket{i}\bra{i}\sigma_{A_1}\ket{i}\bra{i}\, 
\ee
for any $\sigma_{A_1}\in \md(A_1)$.
For multi-partite systems, the preferred basis is the tensor product of the preferred  basis of each subsystem\cite{BCA2015, SSD+2015, WY2016}.

From the golden rule of QRT, the free operations are the set of channels that take the set of incoherent states to itself in the complete sense, i.e., they are completey resource non-generating.
Such operations are called incoherent operations.
In literature, several types of incoherent operations have been studied.
The largest set of incoherent operations is known as the maximally incoherent operations (MIO) \cite{A2006}.
Other incoherent operations include incoherent operations (IO) \citep{BCP2014}, dephasing-covariant incoherent operations (DIO) \cite{CG2016a, CG2016b, CG2016c, MS2016}, strictly incoherent operations (SIO) \cite{WY2016,YMG+2016}, physically incoherent operations (PIO) \cite{CG2016a, CG2016b, CG2016c}, translationally-invariant  operations (TIO) \cite{CY2017}, genuinely incoherent operations (GIO) \cite{VS2016}, fully incoherent operations (FIO) \cite{VS2016}, etc.
In this section, we will briefly discuss about MIO, DIO, IO, and SIO, as we will be defining four sets of free superchannels in the next section taking their analogy.

The maximally incoherent operations (or MIO) \cite{A2006} are defined as the set of CPTP and non-selective maps $\mE \in \ml(A_0 \to A_1)$ such that
\be 
\mE(\rho_{A_0}) \in \mI_{A_1}\quad \forall\; \rho_{A_0} \in \mI_{A_0}
\ee 
Let us denote the set of all channels that follow the above property by $\mio(A_0\to A_1)$.
Any CPTP map $\mM_{A_0 \to A_1} \in \mio(A_0 \to A_1)$ can be characterized using the dephasing channels in the following way
\be
\mM_{A_0 \to A_1} \in \mio(A_0\to A_1) \iff \mD_{A_1} \circ \mM_{A_0\to A_1}\circ \mD_{A_0} = \mM_{A_0 \to A_1 }\circ \mD_{A_0}
\ee
Despite the fact that MIO cannot create coherence, these operations do not have a free dilation, i.e., they cost coherence to be implemented \cite{CG2016a, CG2016b, CG2016c, MS2016}.

A smaller class of free operations, the incoherent operations (or IO) \cite{BCP2014} are defined as the set of CPTP maps $\mE \in \cptp(A_0 \to A_1)$ having a Kraus operator representation $\{K_n \}$ such that
\be
\frac{K_n \rho_{A_0} K_n^{\dagger}}{\tr[K_n \rho_{A_0} K_n^{\dagger}]} \in \mI_{A_1} \quad \forall \; n \text{ and } \rho_{A_0} \in \mI_{A_0}
\ee
This class of operations also do not have a free dilation \cite{CG2016a, CG2016b, CG2016c, MS2016}.

The next class of free operations, the strictly incoherent operations (or SIO) \cite{WY2016,YMG+2016} are defined as the set of CPTP maps $\mE \in \cptp(A_0 \to A_1)$ having a Kraus operator representation $\{K_n \}$ such that
\be
K_n \mD_{A_0}\left(\rho_{A_0}\right) K_n^{\dagger} = \mD_{A_1}\left(K_n \rho_{A_0} K_n^{\dagger}\right) \quad \forall \; n
\ee
This class of operations also do not have a free dilation \cite{CG2016a, CG2016b, CG2016c}.

The last class of free operations that is useful to us is the dephasing-covariant incoherent operations (or DIO) \cite{CG2016a, CG2016b, CG2016c, MS2016}. 
A CPTP map $\mE_{A}$ is said to be DIO if
\be
\left[\mD,\mE_A \right] = 0
\ee
which is equivalent to
\be
\mD_{A_1}\left(\mE_{A_0\to A_1}\left(\rho_{A_0}\right)  \right) = \mE_{A_0\to A_1}\left(\mD_{A_0}\left(\rho_{A_0} \right) \right) \quad \forall \; \rho_{A_0} \in \md(A_0)
\ee

\subsection{Max-relative entropy for channels}

The max-relative entropy is defined on a pair $(\rho, \sigma)$ with $\rho \in \md(A_1)$ and $\sigma \in \pos(A_1)$
 of a state $\rho$ with respect to a positive operator  $\sigma$ is given by
\be
	D_{\max}(\rho \| \sigma) \eqdef \log \min \left\{ t : t\sigma \geq \rho \right\}
\ee
where the inequality sign means that the difference between l.h.s. and r.h.s. is a positive operator.
Similarly for channels, the maximum relative entropy between two $\cp$ maps $\mN$ and $\mE$ is given by
\be\label{dmax_channels}
	D_{\max}(\mN_A \| \mE_A) \eqdef \log \min \left\{ t : t\mE_A \geq \mN_A \right\}
\ee
where the inequality  sign means that the difference between l.h.s. and r.h.s. is a $\cp$ map.
Denoting the Choi matrix of $t\mE_A$ by $\omega_A$, we can write
 \be\label{dmax_channels_Choi}
D_{\max}(\mN_A \| \mE_A) = \min \left\{ \frac{\tr[\omega_A]}{|A_0|} : \omega_A \geq 0\, , \, \omega_{A_0} = I_{A_0}\, , \, \omega_A \geq J^{\mN}_A \right\} 
 \ee
The channel max-relative entropy ($D_{\max}(\mN_A \| \mE_A)$) can be expressed in a simple closed form as a function of the Choi matrices of the maps $\mN_A$ and $\mE_A$ \cite{BHK+2018, FFR+2019}. 
This implies that it is also additive under tensor products.
For completeness, we give the following proof.
\begin{lemma}
 The max-relative entropy for channels is additive under tensor product, i.e.,
 \be
	D_{\max}(\mN_A \otimes \mM_{A'}\| \mE_A \otimes \mF_{A'}) = D_{\max}(\mN_A \| \mE_A) + D_{\max}(\mM_{A'} \| \mF_{A'})
 \ee
\end{lemma}
\begin{proof}
For the proof of the inequality $D_{\max}(\mN_A \otimes \mM_{A'}\| \mE_A \otimes \mF_{A'}) \leq D_{\max}(\mN_A \| \mE_A) + D_{\max}(\mM_{A'} \| \mF_{A'})$, let 
\begin{align}
D_{\max}(\mN_A \| \mE_A) &= \log  \{ t_1 : t_1 \mE_A \geq \mN_A \}\, , \\
D_{\max}(\mM_{A'} \| \mF_{A'}) &= \log \{ t_2 : t_2 \mF_{A'} \geq \mM_{A'} \}\, .
\end{align}
We can rewrite $D_{\max}(\mN_A \otimes \mM_{A'} \| \mE_A \otimes \mF_{A'})$ as
\ba
D_{\max}(\mN_A \otimes \mM_{A'} \| \mE_A \otimes \mF_{A'}) &= \log \min \{ t : t \left(\mE_A \otimes \mF_{A'}\right) \geq \mN_A \otimes \mM_{A'} \} \label{dmax_additive} \\
		& = \log \min \{ t : \frac{t}{t_1 t_2} \left(t_1 \mE_A \otimes t_2\mF_{A'}\right) \geq \mN_A \otimes \mM_{A'} \} 
\ea
From this, we can clearly see 
\be
	\log \min \left\{ t : t \left(\mE_A \otimes \mF_{A'}\right) \geq \mN_A \otimes \mM_{A'} \right\} \leq \log(t_1 t_2) 
\ee
Hence,
\be \label{dmax_leq_condition}
	D_{\max}(\mN_A \otimes \mM_{A'}\| \mE_A \otimes \mF_{A'}) \leq D_{\max}(\mN_A \| \mE_A) + D_{\max}(\mM_{A'} \| \mF_{A'})
\ee
For the proof of $D_{\max}(\mN_A \otimes \mM_{A'}\| \mE_A \otimes \mF_{A'}) \geq D_{\max}(\mN_A \| \mE_A) + D_{\max}(\mM_{A'} \| \mF_{A'})$, note that $D_{\max}$ in \eqref{dmax_channels} and \eqref{dmax_channels_Choi} can be computed using an SDP and its dual is given by
\be
	D_{\max}(\mN_A \| \mE_A) = \log \max \left\{\tr[\beta_A J^{\mN}_A] : \frac{I_A}{|A_0|} + \tau_{A_0}\otimes I_{A_1} \geq \beta_A \right\}
\ee
where $\beta_A \geq 0$ and $\tau_{A_0} \in \herm(A_0)$ such that $\tr[\tau_{A_0}] = 0$.
We can rewrite this as 
\be
	D_{\max}(\mN_A \| \mE_A) = \log \max \left\{\tr[\eta_A J^{\mN}_A] : \eta_A = \gamma_{A_0} \otimes I_{A_1} \, , \, \gamma_{A_0} \geq 0 \, , \, \tr[\gamma_{A_0}] = 1 \right\}
\ee
Now let
\begin{align}
	2^{D_{\max}(\mN_A \| \mE_A)} &= \tr[\eta^1_A J^{\mN}_A] \, , \\
	2^{D_{\max}(\mM_{A'} \| \mF_{A'})} &= \tr[\eta^2_A J^{\mM}_{A'}] \, .
\end{align}
We can write $2^{D_{\max}(\mN_A \otimes \mM_{A'} \| \mE_A \otimes \mF_{A'})} $ as
\ba
	2^{D_{\max}(\mN_A \otimes \mM_{A'} \| \mE_A \otimes \mF_{A'})} &= \max \left\{ \tr\left[\eta_{A A'}\left(J^{\mN}_A \otimes J^{\mM}_{A'} \right)\right] : \eta_{AA'} = (u_{A_0A_0'} + \tau_{A_0 A_0'})\otimes I_{A_1 A_1'} \right\}
\ea
where $\tr[\tau_{A_0 A_0'}] = 0$.
Since the choice of $\eta_{AA'} = \eta^1_A \otimes \eta^2_{A'}$ satisfies the above constraint, therefore we can say
\be
	2^{D_{\max}(\mN_A \otimes \mM_{A'} \| \mE_A \otimes \mF_{A'})} \geq 2^{D_{\max}(\mN_A \| \mE_A)}2^{D_{\max}(\mM_{A'} \| \mF_{A'})}
\ee
which implies
\be \label{dmax_geq_condition}
	D_{\max}(\mN_A \otimes \mM_{A'} \| \mE_A \otimes \mF_{A'}) \geq D_{\max}(\mN_A \| \mE_A) + D_{\max}(\mM_{A'} \| \mF_{A'})
\ee

From \eqref{dmax_leq_condition} and \eqref{dmax_geq_condition}, we can conclude that the max rel-entropy for channels is additive under tensor products, i.e., $D_{\max}(\mN_A \otimes \mM_{A'} \| \mE_A \otimes \mF_{A'}) = D_{\max}(\mN_A \| \mE_A) + D_{\max}(\mM_{A'} \| \mF_{A'})$.
\end{proof}
Apart from this, the $\epsilon$-smooth max-relative entropy is defined and discussed in detail in \cite{LW2019, DKW+2018, FWT+2018}
\be
 D^{\epsilon}_{\max}(\mN_A \| \mM_A) \eqdef \inf_{\mN_A' \in B_{\epsilon}(\mN_A)} D_{\max}(\mN_A' \| \mM_A)
\ee
where 
\be
B_{\epsilon}(\mN_A) = \left\{ \mN_A' \in \cptp(A_0 \to A_1): \frac{1}{2}\|\mN_A' - \mN_A \|_{\diamond} \leq \epsilon \right\}
\ee

\section{The set of free superchannels}\label{free_superchannels}

As discussed in the introduction, the set of free channels in the theory of dynamical coherence are classical channels.
Therefore, a free superchannel consists of a pre-processing classical channel and a post-processing classical channel (see Fig.~\ref{prepost}). However, such a free superchannel always destroy completely any resource; that is, it converts all channels (even coherent ones) into classical channels. This means that the resource theory is in a sense ``degenerate" and no interesting consequences can be concluded from such a theory. 

\begin{figure}[h]
   \includegraphics[width=0.35\textwidth]{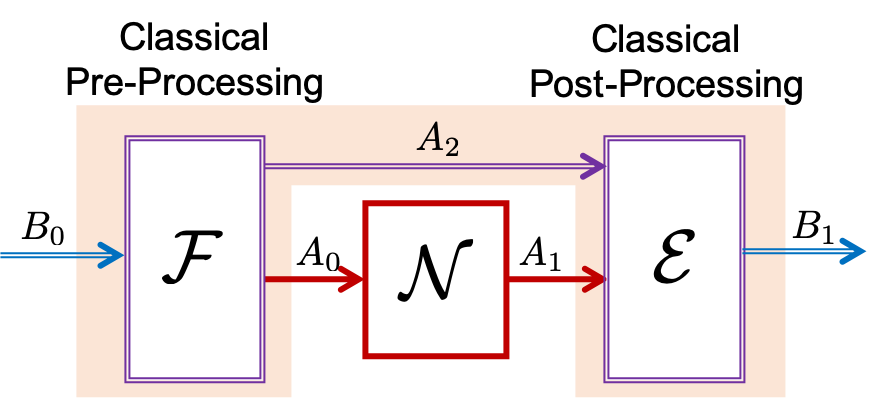}
  \caption{\linespread{1}\selectfont{\small The action of a classical superchannel on a quantum channel.}} 
  \label{prepost}
\end{figure}

This above type of degeneracy also occurs with the resource theory of coherence in the \emph{state domain}. There, the only free operations that are physically consistent are PIO,
which are very restricted and cannot provide much insight into the phenomenon of coherence in quantum systems. Therefore, almost all the enormous amount of work in recent years on the QRT of coherence was devoted to the study of coherence under much larger sets of operations, such as MIO, DIO, IO, and SIO. While these larger sets of operations
cannot be implemented without a coherence cost, they do not generate coherence, and as such they can be used for the study of coherence of states. However, since MIO, DIO, IO, and SIO, all have a coherence cost, they cannot be used as the ``free operations" in a resource theory that aims to quantify the coherence of quantum channels. 

Instead, for a dynamical QRT of coherence, one can define free superchannels that form a larger set than classical superchannels. 
Similar to what happens in the state domain, there is a coherent cost to implement such superchannels, however, they do not generate dynamical coherence, and therefore can be used in a dynamical resource theory of coherence. As it happens in the state domain, there are several natural sets of free superchannels that we can define. 

\subsection{Maximally Incoherent Superchannels (MISC)}

In any quantum resource theory, free operations cannot generate a resource. Taking this principle to the level of superchannels, we define the maximal incoherent superchannels (MISC) as follows.

\begin{definition}
Given two dynamical systems $A$ and $B$, a superchannel $\Theta\in\ms(A\to B)$ is said to be MISC if
\be
\Theta_{A\to B}[\mN_A]\in\mc(B_0\to B_1)\quad\forall\mN_A\in\mc(A_0\to A_1)\;.
\ee
We denote the set of all superchannels that have the above property by $\misc(A\to B)$.
\end{definition}

\begin{remark}
Similar to the characterization of MIO channels with the dephasing channel, the condition that $\Theta$ is in $\misc(A\to B)$ can be characterized with the dephasing superchannel $\Delta_A$. Specifically, we have that
\be\label{defmisc}
\Theta\in\misc(A\to B)\iff \Delta_B\circ\Theta_{A\to B}\circ\Delta_A=\Theta_{A\to B}\circ\Delta_A\;.
\ee
\end{remark}

One of the key properties of any resource theory is that the free operations are ``completely free". This is a physical requirement that a free channel (or superchannel) can act on a subsystem. In the following theorem we show that $\misc(A\to B)$ is completely free. That is, in the QRT we consider here there is no difference between RNG and completely RNG.

\begin{theorem}
Let $A$ and $B$ be two dynamical systems, and let $\Theta\in\misc(A\to B)$. Then, for any dynamical system $R$, the superchannel $\1_R\otimes\Theta$ is free; i.e.
$\1_R\otimes\Theta\in\misc(RA\to RB)$.
\end{theorem}

\begin{proof}
Let $\mN_{RA}\in\mc(R_0A_0\to R_1A_1)$ be a classical channel satisfying 
\be
\Delta_{RA}\left[\mN_{RA}\right]=\Delta_R\otimes\Delta_A\left[\mN_{RA}\right]=\mN_{RA}\;.
\ee
Then,
\begin{align}
\Delta_{RB}\circ\left(\1_R\otimes\Theta_{A\to B}\right)\left[\mN_{RA}\right]&=
\Delta_R\otimes\left(\Delta_B\circ\Theta_{A\to B}\right)\left[\mN_{RA}\right]\\
&=\1_R\otimes\left(\Delta_B\circ\Theta_{A\to B}\right)\left[\mN_{RA}\right]\\
&=\1_R\otimes\left(\Delta_B\circ\Theta_{A\to B}\circ\Delta_A\right)\left[\mN_{RA}\right]\\
&=\1_R\otimes\left(\Theta_{A\to B}\circ\Delta_A\right)\left[\mN_{RA}\right]\\
&=\1_R\otimes\Theta_{A\to B}\left[\mN_{RA}\right]
\end{align}
where the first equality follows from the equality $\Delta_{RA}=\Delta_R\otimes\Delta_A$, the second equality from the fact that $\mN_{RA}$ is classical and in particular $\Delta_R\left[\mN_{RA}\right]=\mN_{RA}$, the third equality from the similar equality $\Delta_A\left[\mN_{RA}\right]=\mN_{RA}$, the fourth equality from~\eqref{defmisc}, and the last equality follows again from $\Delta_A\left[\mN_{RA}\right]=\mN_{RA}$. Hence, $\1_R\otimes\Theta_{A\to B}\left[\mN_{RA}\right]$ is classical so that $\1_R\otimes\Theta\in\misc(RA\to RB)$. This completes the proof.
\end{proof}

The theorem above indicates that MISC can be viewed as the set of completely resource non-generating superchannels in the theory of dynamical coherence. We next consider the characterization of the set MISC. Recall that in the state domain, we can determine if a channel $\mE_A$ belong to MIO$(A_0\to A_1)$ simply by checking if all the states $\mE_A(|x\lr x|_{A_0})$ are diagonal for all $x=1,...,|A_0|$. This simplicity of MIO implies that all state conversions in the single-shot regime can be determined with SDP. In the channel domain, however, the characterization of MISC is slightly more complex.

Recall that the Choi matrix of any classical channel $\mN\in\cptp(A_0\to A_1)$ is a column stochastic matrix. The set of all extreme points (i.e. classical channels) of the set of $|A_0|\times |A_1|$ column stochastic matrices consists of matrices that in each column has $|A_0|-1$ zeros and 1 one. Therefore the number of extreme points is given by $|A_0|^{|A_1|}$. This may give the impression that in order to check if $\Theta\in\misc(A_0\to A_1)$ one has to check if the channel $\Theta[\mE_A]$ is classical for all the  $|A_0|^{|A_1|}$ extreme classical channels. Since the number of conditions is exponential in $|A_1|$ it may give the impression that the problem of deciding if a superchannel belongs to MISC cannot be solved with SDP. However, we show now that this problem can be solved with polynomial (in $|A_0 A_1|$) number of constraints. It can be seen from the relationship between the Choi matrix of $\Theta_{A\to B}$ and that of $\Theta_{A\to B}\circ\Delta_A$ and $\Delta_B\circ\Theta_{A\to B}$.

\begin{lemma}\label{cute}
Let $A$ and $B$ be two dynamical systems, $\Theta\in\ms(A\to B)$ be a superchannel, and $\Delta_A\in\ms(A\to A)$ and $\Delta_B\in\ms(B\to B)$ be the completely dephasing superchannels. Then, the Choi matrices of $\Theta_{A\to B}$, $\Theta_{A\to B}\circ\Delta_A$, and $\Delta_B\circ\Theta_{A\to B}$, satisfy the relations 
\be
\J_{AB}^{\Theta\circ\Delta_A}=\mD_A\left(\J_{AB}^{\Theta}\right)\quad
\text{and}\quad\J_{AB}^{\Delta_B\circ\Theta}=\mD_B\left(\J_{AB}^{\Theta}\right)
\ee
\end{lemma}
\begin{proof}
The Choi matrix of a superchannel $\Theta$ can be expressed as the Choi matrix of the bipartite channel $\Theta_{\tA \to B}\big[ \mP^+_{A\tA} \big]$ \citep{G2019}. 
Similarly, the Choi matrix of the superchannel $\Theta \circ \Delta_A$ can be expressed as the Choi matrix of the bipartite channel $\Theta_{\tA \to B}\circ\Delta_{\tA}\big[ \mP^+_{A\tA}  \big]$
 and that of the superchannel $\Delta_{B}\circ\Theta$ as the Choi matrix of $\Delta_{B}\circ \Theta_{\tA \to B}\big[\mP^+_{A\tA} \big]$.

Denoting $\Theta_{\tA \to B}\big[\mP^+_{A\tA} \big]$ as $\mN_{AB}$, the Choi matrix of the superchannel $\Delta_B\circ\Theta_{A\to B}$ can be written as
\begin{align}
\J^{\Delta_B \circ \Theta}_{AB} &= J^{\Delta_B \left[ \mN_{AB}\right]}_{AB} \\
											   &= \mD_{B_1}\circ \mN_{\tA_0 \tB_0 \to A_1 B_1} \circ \mD_{\tB_0}\left(  \phi^+_{A_0\tA_0} \otimes \phi^+_{B_0\tB_0} \right) \label{Choi_Delta_Theta}
\end{align}
Now using the fact that $\mM_{\tR_0 \to R_1} \ket{\phi^+_{R_0 \tR_0}} = \mM^T_{\tR_1 \to R_0}\ket{\phi^+_{\tR_1 R_1}}$ , we can rewrite  \eqref{Choi_Delta_Theta} as 
 \begin{align}
   \J^{\Delta_B \circ \Theta}_{AB} &= \left( \mD_{B_0} \otimes \mD_{B_1}\right)\circ \mN_{\tA_0 \tB_0 \to A_1 B_1} \left(  \phi^+_{A_0\tA_0} \otimes \phi^+_{B_0\tB_0} \right) \\
   												  &= \mD_B\left(\J^{\Theta}_{AB} \right) 
 \end{align} 
To find $\J^{\Theta \circ \Delta_A}$, note that for any superchannel $\Omega \in \ms(A \to B)$ we have \cite{G2019}
\be\label{superchannel_transpose}
\1_A \otimes \Omega_{\tA \to B}[\mP^+_{A\tA}] = \Omega^T_{\tB \to A} \otimes \1_B[\mP^+_{\tB B}]
\ee
From this, it can be calculated that for the dephasing superchannel,  $\Delta^T = \Delta$.
Therefore, we have
\begin{align}
	\Theta_{\tA \to B} \circ \Delta_{\tA}\left[ \mP^+_{A\tA} \right] &= \Theta_{\tA \to B} \circ \Delta^T_{A}\left[ \mP^+_{A\tA} \right] \\
																								&= \Theta_{\tA \to B} \circ \Delta_{A}\left[ \mP^+_{A\tA} \right] \\
																								&= \Delta_A \circ \Theta_{\tA \to B}\left[ \mP^+_{A\tA} \right] \\
																								&= \Delta_A \circ \mN_{AB}
\end{align}	
So, the Choi matrix of $\Theta_{\tA \to B} \circ \Delta_{\tA}\big[ \mP^+_{A\tA} \big]$ is equal to finding the Choi matrix of $\Delta_A \circ \mN_{AB}$ .
From the calculation of the Choi matrix of $\Delta_B \circ \mN_{AB}$ above, we can easily conclude that
\be
	\J^{\Theta\circ\Delta_A}_{AB} = \mD_A\left(\J^{\Theta}_{AB} \right)
\ee
\end{proof}
With this lemma at hand we get the following characterization for the set $\misc(A\to B)$.
\begin{theorem}
Let $A$ and $B$ be two dynamical systems, and $\Theta\in\ms(A\to B)$ be a superchannel. Then, $\Theta\in\misc(A\to B)$ if and only if
\be\label{thmmisc}
\mD_{AB}\left(\J_{AB}^{\Theta}\right)=\mD_A\otimes\id_B\left(\J_{AB}^{\Theta}\right)\;.
\ee
\end{theorem}
\begin{proof}
From~\eqref{defmisc} and the lemma above we have that
\be
\J_{AB}^{\Theta\circ\Delta_A}=\mD_A\otimes\id_B\left(\J_{AB}^{\Theta}\right)
\ee
is equal to
\be
\J_{AB}^{\Delta_B\circ\Theta\circ\Delta_A}=\id_A\otimes\mD_B\left(\J_{AB}^{\Theta\circ\Delta_A}\right)=\mD_{AB}\left(\J_{AB}^{\Theta}\right)
\ee
This completes the proof.
\end{proof}

Note that for any Hermitian matrix $Z_{AB}\in\herm(AB)$ we have
\be
\tr\left[\Big(\mD_{AB}\left(\J_{AB}^{\Theta}\right)-\mD_A\otimes\id_B\left(\J_{AB}^{\Theta}\right)\Big)Z_{AB}\right]
=\tr\left[\J_{AB}^{\Theta}\Big(\mD_{AB}\left(Z_{AB}\right)-\mD_A\otimes\id_B\left(Z_{AB}\right)\Big)\right]
\ee
Therefore, the theorem above implies that $\Theta\in\misc(A\to B)$ if and only if 
\be\label{xa}
\tr\left[\J^\Theta_{AB}X_{AB}\right]=0\quad\forall X_{AB}\in\mk_\misc
\ee
where $\mk_\misc$ is a subspace of $\herm(AB)$ defined as
\be\label{k_MISC}
\mk_\misc\eqdef\Big\{\mD_{AB}\left(Z_{AB}\right)-\mD_A\otimes\id_B\left(Z_{AB}\right)\;:\;Z_{AB}\in\herm(AB)\Big\}\;.
\ee
Since the dimension of the subspace $\mk_\misc$ is $|AB|(|B|-1)$, it is sufficient to restrict $X_{AB}$ in~\eqref{xa} to the $|AB|(|B|-1)$ elements of some fixed basis of $\mk_\misc$. 
Note also that the condition above is equivalent to the inclusion $\J^\Theta_{AB}\in\mk^\perp_\misc$, where $\mk^\perp_\misc$ is the orthogonal complement of $\mk_\misc$ in $\herm(AB)$.

\subsection{Dephasing Incoherent Superchannels (DISC)}

In the QRT of static coherence, the dephasing channel plays a major role, and in particular, leading to the definition of DIO. Here, the dephasing superchannel defined by $\Delta_A[\mN_A]=\mD_{A_1}\circ\mN_A\circ\mD_{A_0}$ plays a similar roll, as we have already seen in the definition of MISC. We use here the dephasing superchannel to define the set of dephasing incoherent superchannels.

\begin{definition}
Let $A$ and $B$ be two dynamical systems, and let $\Theta\in\ms(A\to B)$ be a superchannel. Then, $\Theta$ is said to be a dephasing incoherent superchannel (DISC) if and only if
\be\label{defdisc}
\Delta_B\circ\Theta_{A\to B}=\Theta_{A\to B}\circ\Delta_A\;.
\ee
Moreover, the set of all such superchannels that satisfy the above relation is denoted by $\disc(A\to B)$.
\end{definition} 

Clearly, from its definition $\disc(A\to B)$ is a subset of $\misc(A\to B)$, and in particular, it is completely free.
Now, from Lemma~\ref{cute} it follows that a superchannel $\Theta\in\disc(A\to B)$ if and only if
\be\label{disc_condition}
\mD_A\otimes\id_B\left(\J_{AB}^{\Theta}\right)=\id_A\otimes\mD_B\left(\J_{AB}^{\Theta}\right)\;.
\ee
Moreover, similar to the considerations above, since the map $\mD_A\otimes\id_B-\id_A\otimes\mD_B$ is self adjoint, it follows that $\Theta\in\disc(A\to B)$ if and only if
\be\label{ya}
\tr\left[\J^\Theta_{AB}Y_{AB}\right]=0\quad\forall Y_{AB}\in\mathfrak{K}_\disc
\ee
where 
\be\label{k_DISC}
\mk_\disc\eqdef\Big\{\id_A\otimes\mD_{B}\left(Z_{AB}\right)-\mD_A\otimes\id_B\left(Z_{AB}\right)\;:\;Z_{AB}\in\herm(AB)\Big\}\;.
\ee
Since the dimension of the subspace $\mk_\disc$ is $|AB|(|A|+|B|-1)$ it is sufficient to restrict $Y_{AB}$ in~\eqref{ya} to the $|AB|(|A|+|B|-1)$ elements of some fixed basis of $\mk_\disc$. 
Note also that the condition above is equivalent to the inclusion $\J^\Theta_{AB}\in\mk^\perp_\disc$, where $\mk^\perp_\disc$ is the orthogonal complement of $\mk_\disc$ in $\herm(AB)$.

\subsection{Incoherent superchannels (ISC) and strictly incoherent superchannels (SISC)}

%As discussed in the preliminary section 
Any superchannel $\Theta\in\ms(A\to B)$ has a Kraus decomposition i.e. an operator sum representation 
\be\label{osr}
\Theta_{A\to B}=\sum_{x=1}^{n}\Theta^x_{A\to B}
\ee
where the Choi matrix of each $\Theta^x_{A\to B}\in\ml(A\to B)$ has rank one. We use this property to define two other sets of free operations that we call incoherent superchannels (ISC) and strictly incoherent superchannels (SISC).

\begin{definition}
Let $A$ and $B$ be two dynamical systems, and let $\Theta\in\ms(A\to B)$ be a superchannel. Then, $\Theta$ is said to be an incoherent superchannel (ISC) if and only if it has a Kraus decomposition $\{\Theta^x_{A\to B}\}_{x=1}^{n}$ as in~\eqref{osr} that satisfies
\be\label{defisc}
\Delta_B\circ\Theta_{A\to B}^x\circ\Delta_A=\Theta_{A\to B}^x\circ\Delta_A\quad\forall\;x=1,...,n.
\ee
Moreover, the set of all such superchannels that satisfy the above relation is denoted by $\isc(A\to B)$.
\end{definition} 

\begin{definition}
Let $A$ and $B$ be two dynamical systems, and let $\Theta\in\ms(A\to B)$ be a superchannel. Then, $\Theta$ is said to be a strictly incoherent superchannel (SISC) if and only if it has a Kraus decomposition $\{\Theta^x_{A\to B}\}_{x=1}^{n}$ as in~\eqref{osr} that satisfies
\be\label{defsisc}
\Delta_B\circ\Theta_{A\to B}^x=\Theta_{A\to B}^x\circ\Delta_A\quad\forall\;x=1,...,n.
\ee
Moreover, the set of all such superchannels that satisfy the above relation is denoted by $\sisc(A\to B)$.
\end{definition} 

\section{Quantification of dynamical coherence}\label{quantification}

In this section, we find the monotones to quantify dynamical coherence.
We also see which relative entropies form a monotone under MISC and DISC.

\subsection{A complete family of monotones}\label{complete_family_of_monotones}
   
%Recently, -entanglement monotones - cite GG-Cm, MMW-XW + others, MMW_XW

In recent works\cite{LW2019, LBL2018, LY2019, GM2019, BDW+2019,  WW2018, PLO+2017, G2019}, various resource measures have been formulated  for a general resource theory of channels and for the dynamical resource theory of entanglement.
A complete set of monotones for both the general resource theory of channels and the resource theory of entanglement of channels was presented in \cite{GM2019}, i.e., it is sufficient to check if all the monotones of this set acting on one channel are greater than the other, then we can convert one channel to the other using the free superchannels of the given resource theory.
It was shown that the complete family of monotones for the dynamical resource theory of NPT entanglement can be computed using an SDP (which otherwise for LOCC-based entanglement is known to be NP-hard \cite{G2003}).

Similarly, we find a complete set of monotones under the free superchannels, $\misc$ and $\disc$. 
In general, for a given quantum resource theory, it is not obvious if these functions are computable, but we show here that for the dynamical resource theory of coherence, these functions can be computed using an SDP.

For a general quantum resource theory, we can define the following complete set of non-negative resource measures for any quantum channel $\mP_B \in \cptp(B_0 \to B_1)$ such that these measures take the value zero on free channels\cite{GM2019} 
\be \label{monotone_general_QRT}
G_{\mP}(\mN_A) \eqdef \max_{\Theta \in \free(A\to B)}\left\langle \mP_B, \Theta\left[\mM_A\right] \right\rangle \; - \max_{\mM_B \in \mg(B_0 \to B_1)}\left\langle \mP_B, \mM_B \right\rangle \quad  \forall \; \mM_A \in \cptp(A_0 \to A_1)\, .
\ee 
where $\mg(B_0\to B_1)$ denotes the set of free channels for the given resource theory.
 
For the dynamical resource theory of coherence, we can define a function $f_{\mP}(\mM_A)$ for any quantum channel $\mP_B \in \cptp(B_0\to B_1)$ and superchannel $\Theta \in \mf(A \to B)$ where $\mf = \misc \text{ or } \disc$, as 
\be \label{f_coherence}
	f_{\mP}(\mM_A) = \max_{\Theta \in \mf(A \to B)} \left\langle \mP_B, \Theta[\mM_B] \right\rangle \quad \forall\; \mM_A \in \cptp(A_0 \to A_1) 
\ee
Note that \eqref{f_coherence} can be expressed as the following SDP for a given $\mM_{A} \in \cptp(A_0 \to A_1)$
\be\label{complete_monotones_SDP}
\max \left\{\tr\left[\J^{\Theta}_{AB}\left(J^{\mM}_A \otimes J^{\mP}_B\right)\right]\right\}
\ee
where the maximum is subject to
\begin{align}
	&\J^{\Theta}_{AB} \geq 0 \, , \, \J^{\Theta}_{AB_0} = \J^{\Theta}_{A_0 B_0} \otimes u_{A_1}\, , \, \J^{\Theta}_{A_1 B_0}	= I_{A_1 B_0} \label{constraint1}\\
	& \tr[\J^{\Theta}_{AB} X^{i}_{AB}] = 0 \; \forall \; i = 1, \ldots , n \label{constraint2}
\end{align} 
where $\{X^i_{AB}\}_{i=1}^n$ can denote the basis of the subspace $\mk_\mf$ as defined in \eqref{k_MISC} and \eqref{k_DISC} for $\mf = \misc \text{ and } \disc$, respectively.
 For $\misc$, $n\equiv|AB|(|B|-1)$ whereas for $\disc$, $n \equiv |AB|(|A| + |B| -1)$. 
Conditions in \eqref{constraint1} are there because $\Theta(A\to B)$ is a superchannel whereas conditions in \eqref{constraint2} are the result of the requirement that $\Theta \in \mf(A\to B)$.

Similar to \eqref{monotone_general_QRT}, for all $\mP \in \cptp(B_0 \to B_1)$, we can define
\be
G_{\mP}(\mN_A) \eqdef \max \tr\left[ \J^{\Theta}_{AB}\left( \left(J^{\mN}_A\right)^T \otimes J^{\mP}_B  \right)\right] - \max \tr[J^{\mM}_B J^{\mP}_B]
\ee
where the second maximum is over all $\mM_B \in \mc(B_0 \to B_1)$ and the first maximum is subject to the constraints given in \eqref{constraint1} and \eqref{constraint2}.
The family $\left\{G_{\mP} \right\}$ over all $\mP \in \cptp(B_0 \to B_1)$ is a complete set of monotones, that is, there exists a $\Theta \in \mf(A\to B)$ where $\mf = \misc \text{ or }\disc$, that can convert a channel $\mN_A \in \cptp(A_0 \to A_1)$ to $\mM_B \in \cptp(B_0 \to B_1)$ if and only if
\be
G_{\mP}(\mN_A) \geq G_{\mP}(\mM_B)
\ee
\begin{remark1}
  For the qubit case we calculated the values of the monotone $G_{\mP}(\mN_A)$ under MISC for a few channels(or a class of channels) by plugging into CVX.
  This required construction of 48 basis elements (Eq. \eqref{k_MISC}).
  The value of $G_{\mP}(\mN_A)$ for all classical channels is 0 for all $\mP_B$.
  We found that for a fixed $\mP_B$, the value of all unitary channels is the same and they attain the maximum value of 2 when $\mP_B$ is the identity channel.
  If we fix $\mP_B$ to be the identity channel, we see that for a replacement channel that outputs a plus state ($\ket{+} = \frac{1}{\sqrt{n}}\sum_{i=0}^{n-1} \ket{i}$), the value of $G_{\id}(\mN_A)$ is equal to 2.
  For any other replacement channel and any depolarizing channel, $G_{\id}(\mN_A)$ is less than 2.
\end{remark1}
\begin{remark2}
 Since there are an infinite number of monotones in the above complete set $G_{\mP}$, it might give an impression that the conversion of a channel $\mN_A\in \cptp(A_0 \to A_1)$ to another channel 
$\mM_B\in \cptp(B_0 \to B_1)$ using a superchannel $\Theta \in \misc \text{ or }\disc$, is very hard or impractical, but in section \ref{interconversion} we show that the problem of interconversion of two quantum channels using a superchannel belonging to $\misc$ or $\disc$ can be computed using an SDP.
\end{remark2}

\subsection{Relative entropies of dynamical coherence}\label{relative_entropies}

%define divergence, R{\'e}nyi divergence - sandwch/quantum and petz. then go to channel divergence. then go to relative entropies(4 non-state ones) as a measure for misc and disc.
A measure of distinguishability or divergence $D(\cdot\| \cdot)$ of two states is a function $D : \md(A_1) \times \md(B_1) \to \mathbb{R}$ such that it obeys data-processing inequality and is zero on the set of free states.
One example of such a function is R{\'e}nyi divergence\cite{R1961}.
 Its two quantum generalizations which have been given an operational interpretation are ``Sandwiched" R{\'e}nyi Relative Entropy (also known as Quantum R{\'e}nyi Divergence) and Petz-R{\'e}nyi relative entropy.
 ``Sandwiched" R{\'e}nyi Relative Entropy (or Quantum R{\'e}nyi Divergence) was introduced and discussed in \cite{WWY2014, MDS+2013, DL2014} whereas Petz-R{\'e}nyi relative entropy was introduced and studied in \cite{HMP+2011, HM2017, W2018a}.
Other generalizations of the R{\'e}nyi divergence and the quantum R{\'e}nyi relative entropies are discussed in \cite{AD2015} but their operational meaning is not clear. 
 
For channels, the relative entropies and divergence have been generalized from the state case (i.e., static resources) to channels (i.e., dynamic resources) and were discussed in \cite{LY2019, LW2019, GW2019, G2019, CMW2016, LKD+2018, BHK+2018}.
%In this section, we take the relative entropies listed in \cite{GW2019} and see which relative entropies are monotonic under MISC and DISC. 
%The following relative entropies form a monotone under $\misc$
We take the relative entropies listed in  \cite{GW2019} and find the following three relative entropies to be clearly forming a monotone under MISC
\begin{align}
	C_1\left(\mN_A\right) &= \min_{\mM\in\mc(A_0\to A_1)} \max_{\phi \in \md(R_0A_0)} D\left(\mN_{A_0 \to A_1}\left(\phi_{R_0A_0}\right)\big\|\mM_{A_0 \to A_1}\left(\phi_{R_0 A_0}\right)\right) \\
	C_2\left(\mN_A\right) &= \min_{\mM\in\mc(A_0\to A_1)} \sup_{\rho, \sigma \in \md(R_0A_0)}D\left(\mN_A\left(\rho_{R_0A_0}\right)\big\|\mM_A\left(\sigma_{R_0A_0}\right) \right) - D\left(\rho_{R_0A_0}\big\|\sigma_{R_0A_0}\right)\\
C_3\left(\mN_A\right) &= \max_{\rho\in \md(R_0A_0)}D\left(\mN_{A}\left( \rho_{R_0A_0} \right) \right) - D\left( \rho_{R_0A_0}  \right)	
\end{align}
where $D(\rho) = \min_{ \mD(\sigma) = \sigma}D\left(\rho \| \sigma  \right)$ and $D(\rho \| \sigma) = \tr[\rho\log \rho - \rho\log \sigma]$ is the relative entropy.
The proof that the above relative entropies form a monotone under $\misc$ is similar to the proof for relative entropies forming a monotone for a general resurce theory of quantum processes as given in \cite{GW2019}.
Note that the relative entropies $C_1(\mN_A)$ and $C_2(\mN_A)$ are faithful, i.e., they take the value zero iff $\mN_A \in \mc(A_0 \to A_1)$. The relative entropy $C_3(\mN_A)$ is a state-based relative entropy and involves no optimization over the classical channels.

In \cite{GW2019}, there are three other relative entropies defined by taking the optimization over the set of free states instead of all density matrices.
There, the proof relies on the pre-processing channel to be completely resource non-generating.
Since, we cannot make this assumption, hence, we cannot say about the monotonicity of the relative entropies where the optimization is over the incoherent states. 

For any channel divergence $D$, define the function $D_\Delta:\cptp\to\mbb{R}_{+}$ given by
\be
D_\Delta(\mN_A)\eqdef D\left(\mN_A\big\|\Delta_A\left[\mN_A\right]\right)
\ee
and for the choice $D=D_{\max}$ we call it the dephasing logarithmic robustness and denote it by $D_\Delta\equiv LR_\Delta$.

\begin{lemma}
The function $D_\Delta$ is a dynamical resource monotones under $\disc$.
\end{lemma}
\begin{proof}
Lets $\Theta\in\disc(A\to B)$ and $\mN\in\cptp(A_0\to A_1)$. Then,
\ba
D_\Delta(\Theta_{A\to B}[\mN_{A}])&=D\left(\Theta_{A\to B}[\mN_A]\big\|\Delta_B\circ\Theta_{A\to B}[\mN_A]\right)\\
&=D\left(\Theta_{A\to B}[\mN_A]\big\|\Theta_{A\to B}\circ\Delta_A\left[\mN_A\right]\right)\\
&\leq D\left(\mN_A\big\|\Delta_A\left[\mN_A\right]\right)\\
&=D_\Delta(\mN_A)\;.
\ea
This completes the proof.
\end{proof}
For the case that $D(\rho\|\sigma)=\tr[\rho\log\rho]-\tr[\rho\log\sigma]$ is the relative entropy, we call $D_\Delta$ the dephasing relative entropy of coherence. 

\subsection{Operational Monotones}\label{operational_monotones}

Here, we discuss the monotones that are operationally meaningful for the resource theory of quantum coherence.
We will see that the monotones which are based on $D_{\max}$, like various types of log-robustness, play a major role in the calculation of coherence cost of channels.

The log-robustness of entanglement for states was introduced and investigated in \cite{D2009a, D2009b, BD2010, BD2011}.
It was shown that it is an entanglement monotone and its operational significance for the manipulation of entanglement was also discussed.
The log-robustness of coherence for states was similarly defined in \cite{BSF+2017} and it was shown that it is a measure of coherence.

The log-robustness of channels for a general resource theory was introduced and discussed in \cite{LW2019, GW2019, LY2019}.
It was shown that the log-robustness of channels satisfy necessary conditions for the resource measure of channels, i.e., it is both faithful and a monotone under left and right compositions\cite{LW2019}.

The log-robustness of coherence of channels is given by
\be
  LR_\mc(\mN_{A}) \eqdef\min_{\mE\in\mc(A_0\to A_1)}
                 D_{\max}\big(\mN_{A}\|\mE_{A}\big)
\ee
It can be computed with an SDP. To see why, note that
\ba\label{LR_as_SDP}
LR_\mc(\mN_{A})&=\log\min\Big\{t\geq 0\;:\;t\mE_A\geq\mN_A\;\;\,\;\;\Delta_A[\mE_A]=\mE_A\;\;,\;\;\mE\in\cptp(A_0\to A_1)\Big\}
\ea
Denoting by $\omega_A$ the Choi matrix of $t\mE_A$ we get that (recall that we are using $u$ to denote the maximally mixed state) 
\ba \label{lr_as_sdp}
LR_\mc(\mN_{A})&=\log\min\Big\{\frac{1}{|A_0|}\tr[\omega_A]\;:\;\omega_A\geq J^\mN_A\;\;\,\;\;\mD_A[\omega_A]=\omega_A\;\;,\;\;\omega_{A_0}=\tr[\omega_A]u_{A_0}\;\;,\;\;\omega_A\geq 0\Big\}
\ea
which is an SDP optimization problem. As such it has a dual given by (see appendix for details)
\ba\label{dual_lr}
LR_\mc(\mN_{A})&=\log\max\Big\{\tr[\eta_AJ^\mN_A]\;:\;\mD_A(\eta_A)=\mD_{A_0}\left(\eta_{A_0}\right)\otimes u_{A_1}\;\;\,\;\;\mD_{A_1}[\eta_{A_1}]=I_{A_1}\;\;,\;\;\eta_A\geq 0\Big\}
\ea
\begin{remark}
For the qubit case, we calculated the log-robustness of coherence of few channels.
For any classical channel, the log-robustness of coherence is equal to 0.
For the identity channel it is equal to 1.
For any replacement channel and depolarizing channel, its value is between 0 and 1.
If the replacement channel is the one that outputs the plus state ($\ket{+} = \frac{1}{\sqrt{n}}\sum_{i=0}^{n-1} \ket{i}$), the log-robustness is equal to 1.
Lastly, for any unitary channel, we found that the value of log-robustness of coherence is between 1 and 2.
\end{remark}

Next, we show the additivity of log-robustness of coherence of channels under tensor products.
\begin{lemma}\label{additivity_of_lr_LEMMA}
The log-robustness of coherence of a channel is additive under tensor products, i.e.,
\begin{equation}\label{additivity_of_lr}
	LR_{\mathfrak{C}}(\mathcal{N}_A \otimes \mathcal{M}_{A'})  = LR_{\mathfrak{C}}(\mathcal{N}_A) + LR_{\mathfrak{C}}(\mathcal{M}_{A'})
\end{equation}

\end{lemma}

\begin{proof}

For the proof of the inequality $ LR_{\mathfrak{C}}(\mathcal{N}_A \otimes \mathcal{M}_{A'}) \leq LR_{\mathfrak{C}}(\mathcal{N}_A) + LR_{\mathfrak{C}}(\mathcal{M}_{A'})$, 
let
 $LR_{\mathfrak{C}}(\mathcal{N}_A) = D_{\max}(\mN_A || \mE_A)$
and
 $LR_{\mathfrak{C}}(\mathcal{M}_{A'}) = D_{\max}(\mM_{A'} || \mE_{A'})$.
Then, we have
\begin{align}
	LR_{\mathfrak{C}}(\mathcal{N}_A \otimes \mathcal{M}_{A'}) &\leq D_{\max}(\mathcal{N}_A \otimes \mathcal{M}_{A'} || \mathcal{E}_A \otimes \mathcal{E}_{A'})\\
	&= D_{\max}(\mN_A || \mE_A) + D_{\max}(\mM_{A'} || \mE_{A'})\\
	&= LR_{\mathfrak{C}}(\mathcal{N}_A) + LR_{\mathfrak{C}}(\mathcal{M}_{A'})
\end{align}

The first inequality follows trivially from the definition of log-robustness and the second equality follows from the additivity of $D_{\max}$.

To prove the converse, i.e., $ LR_{\mathfrak{C}}(\mathcal{N}_A \otimes \mathcal{M}_{A'}) \geq LR_{\mathfrak{C}}(\mathcal{N}_A) + LR_{\mathfrak{C}}(\mathcal{M}_{A'})$, we will use the dual of the log-robustness as given in Eq.(\ref{dual_lr}).
Let $\eta_A$ and $\eta_{A'}$ be the optimal matrices for the dual of $LR_{\mc}(\mN_A)$ and $LR_{\mc}(\mM_{A'})$, respectively.
We get
\begin{equation}
    \begin{split}
        & 2^{LR_{\mathfrak{C}}(\mathcal{N}_A)} = \frac{1}{|A_0|} \tr[\eta_A J^{\mathcal{N}_A}_A]\\
        & 2^{LR_{\mathfrak{C}}(\mathcal{M}_{A'})} = \frac{1}{|{A'}_0|} \tr[\eta_{A'}J^{\mathcal{M}_{A'}}_{A'}]
    \end{split}
\end{equation}

Since,  $LR_{\mc}(\mN_A \otimes \mM_{A'}) = \frac{1}{|A_0 A'_0|} \log \max \tr \big[\eta'_{AA'} \big(J^{\mN_A \otimes \mM_{A'}}_{AA'}\big)\big] $ where the maximum is over all $\eta'_{AA'} \geq 0$ satisfying
\be
\mD_{AA'}(\eta'_{AA'})=\mD_{A_0A_0'}\left(\eta'_{A_0A_0'}\right)\otimes u_{A_1A_1'}\;\; ,\;\;\mD_{A_1A_1'}[\eta'_{A_1A_1'}]=I_{A_1A_1'}\;.
\ee
and because $\eta_{AA'} = \eta_A \otimes \eta_{A'}$ satisfies the above conditions, we have
\begin{equation}\label{lr_tensor_prod}
\begin{split}
     2^{LR_{\mathfrak{C}}(\mathcal{N}_A \otimes \mathcal{M}_{A'})} &\geq 
     \frac{1}{|A_0 {A'}_0|}\tr\big[\eta_{A A'}\big(J^{\mathcal{N}_A \otimes \mathcal{M}_{A'}}_{A A'}\big)\big]\\
    &= 2^{LR_{\mathfrak{C}}(\mathcal{N}_A)}2^{LR_{\mathfrak{C}}(\mathcal{M}_{A'})}
\end{split}
\end{equation}
Hence, the above equation implies
\begin{equation}
    LR_{\mathfrak{C}}(\mathcal{N}_A \otimes \mathcal{M}_{A'}) \geq  LR_{\mathfrak{C}}(\mathcal{N}_A) + LR_{\mathfrak{C}}(\mathcal{M}_{A'})
\end{equation}
This establishes the additivity of the log-robustness of a quantum channel, i.e., $LR_{\mathfrak{C}}(\mathcal{N}_A \otimes \mathcal{M}_{A'}) = LR_{\mathfrak{C}}(\mathcal{N}_A) + LR_{\mathfrak{C}}(\mathcal{M}_{A'})$
\end{proof}

Another type of log-robustness, the dephasing logarithmic robustness, which will be used to find the exact cost under DISC, is defined by
\be \label{dephasing_lr}
LR_\Delta(\mN_A)\eqdef D_{\max}\big(\mN_A\big\|\Delta_A[\mN_A]\big)\quad\forall\;\mN\in\cptp(A_0\to A_1)\;.
\ee
We prove here that the dephasing log-robustness is also additive.
\begin{lemma}
Let $\mN\in\cptp(A_0\to A_1)$ and $\mM\in\cptp(B_0\to B_1)$ be two channels. Then,
\be
LR_\Delta\big(\mN_A\otimes\mM_B\big)=LR_\Delta(\mN_A)+LR_\Delta(\mM_B)\;.
\ee
\end{lemma}
\begin{proof}
\ba
LR_\Delta\big(\mN_A\otimes\mM_B\big)&=D_{\max}\big(\mN_A\otimes \mM_B\big\|\Delta_{AB}\big[\mN_A\otimes\mM_B\big]\big)\\
&=D_{\max}\big(\mN_A\otimes \mM_B\big\|\Delta_{A}\big[\mN_A\big]\otimes\Delta_{B}\big[\mM_B\big]\big)\\
&=D_{\max}\big(\mN_A\big\|\Delta_{A}\big[\mN_A\big]\big)+D_{\max}\big(\mM_B\big\|\Delta_{B}\big[\mM_B\big]\big)\\
&=LR_\Delta(\mN_A)+LR_\Delta(\mM_B)\;,
\ea
where the third equality follows from the additivity of $D_{\max}$ for channels.
\end{proof}

We also define smoothed logarithmic robustness and asymptotic logarithmic robustness.
From \cite{GW2019}, we know that smoothing maintains monotonicity.
The smoothed logarithmic robustness is defined by
\be
LR^{\epsilon}_\mc(\mN_A) 
\eqdef \min_{\mN'\in B_\epsilon(\mN_A)}LR_\mc(\mN_A')
\ee
where
\be
B_\epsilon(\mN_A)=\Big\{\mN'\in\cptp(A_0\to A_1)\;:\;\frac{1}{2}\|\mN'_{A}-\mN_{A}\|_{\diamond}\leq \epsilon\Big\}\;.
\ee
and the asymptotic logarithmic robustness is defined as
\be
LR_{\mc}^{\infty}(\mN_{A})=\lim_{\epsilon\to 0^+}\liminf_{n\to\infty}\frac{1}{n}LR^{\epsilon}_{\mc}(\mN_{A}^{\otimes n})
\ee

Similarly we define the smoothed dephasing logarithmic robustness and asymptotic dephasing logarithmic robustness.
The smoothed dephasing logarithmic robustness is defined by
\be
LR^{\epsilon}_\Delta(\mN_A) 
\eqdef \min_{\mN'\in B_\epsilon(\mN_A)}LR_\Delta(\mN_A')
\ee
and the asymptotic dephasing logarithmic robustness as
\be
LR_{\Delta}^{\infty}(\mN_{A})=\lim_{\epsilon\to 0^+}\lim_{n\to\infty}\frac{1}{n}LR^{\epsilon}_{\Delta}(\mN_{A}^{\otimes n})
\ee
Now we define the log-robustness with ``liberal'' smoothing \cite{GW2019} which we find to have an operational meaning. 
Let
\be
LR^{\epsilon,\varphi}_\mc(\mN_A) \eqdef\min_{\mN'\in B_\epsilon^{\varphi}(\mN_A)}
                 LR_\mc\big(\mN'_{A}\big)\;.
\ee
where
\be 
B_\epsilon^{\varphi}(\mN_A)\eqdef\Big\{\mN'\in\text{CP}(A_0\to A_1)\;:\;\|\mN'_{A}(\varphi_{RA_0})-\mN_{A}(\varphi_{RA_0})\|_{1}\leq \epsilon\Big\}.
\ee
and consider its ``liberal smoothing"
\be\label{liberal}
  LR^{\epsilon}_\mc(\mN_A) 
\eqdef \max_{\varphi\in\mD(RA_0)}LR^{\epsilon,\varphi}_\mc(\mN_A).
\ee
Define also
\be
  LR^{\epsilon,n}_\mc(\mN_A) 
\eqdef\frac{1}{n} \max_{\varphi\in\mD(RA_0)}LR^{\epsilon,\varphi^{\otimes n}}_\mc(\mN^{\otimes n}_A)\;,
\ee
and
\be
LR^{(\infty)}_\mc(\mN_A)\eqdef\lim_{\epsilon\to 0^+}\liminf_{n\to\infty}LR^{\epsilon,n}_\mc(\mN_A)\;.
\ee
In \cite{GW2019}, a new type of regularized relative entropy of a resource given by
\be
D^{(\infty)}_\mc\left(\mN_{A}\right)\eqdef\lim_{n\to\infty}\frac{1}{n}\sup_{\varphi\in\md(RA_0)}\min_{\mE\in\mc(A^n_0\to A^n_1)}
D\left(\mN^{\otimes n}_{A_0\to A_1}\left(\varphi_{RA_0}^{\otimes n}\right)
\big\|\mE_{A^n_0\to A^n_1}\left(\varphi_{RA_0}^{\otimes n}\right)\right)
\ee
The quantity $D^{(\infty)}_\mc\left(\mN_{A}\right)$ behaves monotonically under completely RNG superchannels and satisfies the following AEP
\be\label{aep}
LR^{(\infty)}_\mc(\mN_A) =D^{(\infty)}_\mc\left(\mN_{A}\right)\;.
\ee

\section{Interconversions}\label{interconversion}

We show that for the dynamical resource theory of coherence, the interconversion distance $d_{\mf}(\mN_A \to \mM_B)$ can be computed with an SDP.
We then calculate the exact, approximate and ``liberal''  coherence cost of a channel and show that the ``liberal'' cost of coherence is equal to a variant of the regularized relative entropy.

\subsection{The conversion distance of coherence}\label{interconversion_as_SDP}

The conversion distance from a channel $\mN_{A}\in\cptp(A_0\to A_1)$ to a channel $\mM_{B}\in\cptp(B_0\to B_1)$ is defined as (with $\mf$ standing for either one of the four operations $\misc$, $\disc$, $\isc$, and $\sisc$)
\be\label{d_N_to_M}
d_{\mf}\left(\mN_{A}\to\mM_B\right)\eqdef\min_{\Theta\in\mf(A\to B)}\frac{1}{2}\left\|\Theta_{A\to B}\left[\mN_A\right]-\mM_B\right\|_\diamond\;.
\ee
That is, if the conversion distance above is very small then $\mN_{A}$ can be used to simulate a channel that is very close to $\mM_B$, using free superchannels. We now show that for $\mf=\misc$ or $\mf=\disc$, this conversion distance can be computed with a semi-definite program (SDP).

\begin{theorem}\label{misc_disc_sdp}
For the case $\mf=\misc$,  let $\{X^i_{AB}\}_{i=1}^n$ be the basis of the subspace $\mk_\mf$ as defined in \eqref{k_MISC} where $n\equiv|AB|(|B|-1)$ and let $\alpha_{AB}$  denote the Choi matrix of the superchannel $\Theta$.
Then, $d_{\mf}\left(\mN_{A}\to\mM_B\right)$, can be expressed as the following SDP
\be
d_{\mf}\left(\mN_{A}\to\mM_B\right) = \min \lambda
\ee
where the minimum is subject to 
\begin{align}
&\lambda I_{B_0} \geq \omega_{B_0}  \; , \; \omega_B \geq 0 \; , \; \alpha_{AB} \geq 0  \; , \; \omega_B \geq \tr_A\Big[\alpha_{AB}\left( (J^{\mN}_A)^T \otimes I_B \right)\Big] - J^{\mM}_B \; , \\
 						 &\alpha_{AB_0} = \alpha_{A_0 B_0} \otimes u_{A_1} \; , \; \alpha_{A_1 B_0} = I_{A_1 B_0}\; ,\\
 						 & \tr[\alpha_{AB}X_{AB}^i] = 0 \; \forall \; i = 1\, , \, \ldots \, , n
\end{align}

For the case $\mf=\disc$,  let $\{Y_{AB}^i\}_{i=1}^m$ be the basis of the subspace $\mk_\mf$ as defined in \eqref{k_DISC} where $m\equiv |AB|(|A| + |B| -1)$ and $\alpha_{AB}$  denote the Choi matrix of the superchannel $\Theta$.
Then, $d_{\mf}\left(\mN_{A}\to\mM_B\right)$, can be expressed as the following SDP
\be
d_{\mf}\left(\mN_{A}\to\mM_B\right) = \min \lambda
\ee
where the minimum is subject to 
\begin{align}
&\lambda I_{B_0} \geq \omega_{B_0}  \; , \; \omega_B \geq 0 \; , \; \alpha_{AB} \geq 0  \; , \; \omega_B \geq \tr_A\Big[\alpha_{AB}\left( (J^{\mN}_A)^T \otimes I_B \right)\Big] - J^{\mM}_B \; , \\
 						 &\alpha_{AB_0} = \alpha_{A_0 B_0} \otimes u_{A_1} \; , \; \alpha_{A_1 B_0} = I_{A_1 B_0}\; ,\\
 						 & \tr[\alpha_{AB}Y_{AB}^i] = 0 \; \forall \; i = 1\, , \, \ldots \, , m
\end{align}
\end{theorem}

\subsection{Exact Asymptotic Coherence Cost}\label{exact_asymptotic_coherence_cost}

The exact single-shot coherence cost is defined for $\mN_A\in\cptp(A_0\to A_1)$ as
\be
C^{0}_{\mf}(\mN_{A})\eqdef\min\left\{\log |R_1|\;:\;\exists\Theta\in\mf(R_1\to A)\quad\text{s.t.}\quad\Theta_{R_1\to A}[\phi^{+}_{R_1}]=\mN_A\right\}\;,
\ee
where we consider the two cases of $\mf=\misc$ and $\mf=\disc$.
And the exact coherence cost is given by
\be
C^{\exact}_{\mf}(\mN_{A})=\lim_{n\to\infty}\frac{1}{n}C^{0}_{\mf}\left(\mN_{A}^{\otimes n}\right)
\ee
We now compute this coherence cost for both MISC and DISC.

\subsubsection{Exact cost under MISC}

\begin{theorem}
For $\mf=\misc$ and $\mN\in\cptp(A_0\to A_1)$,
\be
C^{\exact}_{\mf}(\mN_{A})=LR_\mc(\mN_{A})
\ee
\end{theorem}

\begin{proof}
We first prove that
\be\label{exact_cost_bw_lr}
 LR_\mc(\mN_A) \leq C^{0}_{\mf}(\mN_A) \leq  LR_\mc(\mN_A) +1
\ee
and then use the additivity of $LR_\mc(\mN_A)$. 

For the proof of $ LR_\mc(\mN) \leq C^{0}_{\mf}(\mN) $, let $\Theta \in \misc(R_1 \to A)$ be a optimal superchannel satisfying $\Theta_{R_1 \to A}[\phi^+_{R_1}] = \mN_A$ such that $ C^{0}_{\misc}(\mN_A) = \log_2 |R_1|$.
Therefore,

\begin{align}
	LR_\mc(\mN_A) &= D_{\max}(\mN_A \big\|\mE_A) \\
				  &= D_{\max}(\Theta_{R_1 \to A}[\phi^+_{R_1}] \big\| \mE_A) \\
				  &\leq D_{\max}(\Theta_{R_1 \to A}[\phi^+_{R_1}]  \big\| \Theta_{R_1 \to A}[\mD(\phi^+_{R_1})]) \\
				&\leq D_{\max}(\phi^+_{R_1}  \big\|  \mD(\phi^+_{R_1}) \;) \\
				&= \log_2 |R_1| \\
				&= C^0_{\mf}(\mN_A)
\end{align}
To prove $C^{0}_{\mf}(\mN_A) \leq  LR_\mc(\mN_A) +1$, first let 
\begin{equation}
	LR_\mc(\mN_A) = D_{\max}(\mN_A \; \| \; \mE_A) = \log_2 t
\end{equation}
 for some optimal $t$ satisfying  $t\mE_A \geq \mN_A$.
 Also, let $m = \ceil{t} $, so that $ m \mE_A \geq \mN_A$ still holds.
 Let $R_1$ be a static system such that $|R_1| = m$.
 We now define the following supermap. For any state $\rho_{R_1} \in \md(R_1)$
 \begin{align}
 	\Omega_{R_1 \to A}[\rho_{R_1}] \eqdef \frac{m}{m-1}\Big( \tr[\phi^+_{R_1} \rho_{R_1}] - \frac{1}{m} 											\Big) \mN_A \; + \; \frac{m}{m-1}\Big( 1 - \tr[\phi^+_{R_1} \rho_{R_1}]  \Big) \mE_A 								 
 \end{align}
Note that the supermap $\Omega_{R_1 \to A} \in \mf(R_1 \to A)$ as it can be expressed as
\be
 \Omega_{R_1 \to A}[\rho_{R_1}] \eqdef \tr[\phi^+_{R_1} \rho_{R_1}]\mN_A \; + \; \frac{1}{m-1}\Big( 1 - \tr[\phi^+_{R_1} \rho_{R_1}]  \Big)(m \mE_A - \mN_A) 	
\ee 
where $m \mE_A - \mN_A \geq 0$.
Also observe that $\Omega_{R_1 \to A}(\phi^+_{R_1}) = \mN_A$.
Hence, such a superchannel implies that
\be
	C^{0}_{\mf}(\mN_A) \; = \; \log_2 m \; = \; \log_2 \ceil{t} \; \leq \; \log_2 t + 1 \; = \; LR_\mc(\mN_A) + 1
\ee
This completes the proof of $LR_\mc(\mN_A) \leq C^{0}_{\mf}(\mN_A) \leq  LR_\mc(\mN_A) +1$.

Therefore, using the additivity of $LR_\mc(\mN_A)$, we can conclude
\be
	C^{\exact}_{\mf}(\mN_A) = LR_\mc(\mN_A)
\ee
\end{proof}

\subsubsection{Exact cost under DISC}

The dephasing logarithmic robustness is given by \eqref{dephasing_lr}
\be
LR_\Delta(\mN_A)\eqdef D_{\max}\big(\mN_A\big\|\Delta_A[\mN_A]\big)\quad\forall\;\mN\in\cptp(A_0\to A_1)\;.
\ee
By definition we have $LR_\mc(\mN_A)\leq LR_\Delta(\mN_A)$. While the logarithmic robustness behaves monotonically under any superchannel in MISC, the dephasing logarithmic robustness is in general not monotonic under MISC. Instead, it is monotonic under DISC.

\begin{lemma}
For any $\mN\in\cptp(A_0\to A_1)$ and $\Theta\in\disc(A\to B)$ we have
\be
LR_\Delta\big(\Theta_{A\to B}[\mN_A]\big)\leq LR_\Delta(\mN_A)\;.
\ee
\end{lemma}
\begin{proof}
\ba
LR_\Delta\big(\Theta_{A\to B}[\mN_A]\big)&=D_{\max}\big(\Theta_{A\to B}[\mN_A]\big\|\Delta_A\circ\Theta_{A\to B}[\mN_A]\big)\\
&=D_{\max}\big(\Theta_{A\to B}[\mN_A]\big\|\Theta_{A\to B}\circ\Delta_A[\mN_A]\big)\\
&\leq D_{\max}\big(\mN_A\big\|\Delta_A[\mN_A]\big)\\
&=LR_\Delta(\mN_A)\;,
\ea
where the second equality follows from the commutativity of $\Theta$ and $\Delta$, and the inequality follows from the data processing inequality of the channel divergence $D_{\max}$~\cite{G2019}.
\end{proof}

\begin{theorem}
For $\mf=\disc$, and $\mN\in\cptp(A_0\to A_1)$
\be
C^{\exact}_{\mf}(\mN_{A})=LR_\Delta(\mN_{A})
\ee
\end{theorem}

\begin{proof}
We first prove that
\be\label{del}
 LR_\Delta(\mN_A) \leq C^{0}_{\disc}(\mN_A) \leq  LR_\Delta(\mN_A) +1
\ee
and then use the additivity of $LR_\Delta$. 

For the proof of $ LR_\Delta(\mN_A) \leq C^{0}_{\disc}(\mN_A) $, let $\Theta\in\disc(R_1\to A)$ be an optimal superchannel satisfying $\Theta_{R_1\to A}[\phi^+_{R_1}] = \mN_A$ such that $ C^{0}_{\disc}(\mN_A) = \log_2 |R_1|$. 
Therefore,

\begin{align}
	LR_\Delta(\mN_A) &=D_{\max}\left(\mN_A \big\|\Delta_A[\mN_A]\right) \\
				&= D_{\max}\left(\Theta_{R_1\to A}[\phi^+_{R_1}] \big\| \Delta_A\circ\Theta_{R_1\to A}\left[\phi^+_{R_1}\right]\right) \\
				&=D_{\max}\left(\Theta_{R_1\to A}[\phi^+_{R_1}] \big\| \Theta_{R_1\to A}\left[\mD_{R_1}(\phi^+_{R_1})\right]\right) \\
				&\leq D_{\max}\left(\phi^+_{R_1} \big\| \mD_{R_1}(\phi^+_{R_1}) \right) \\
				&= \log_2 |R_1|\\
				&=C^{0}_{\disc}(\mN_A)\;.
\end{align}

For the proof of $C^{0}_{\disc}(\mN_A) \leq  LR_\Delta(\mN_A) +1$, first let 
\be
	LR_\Delta(\mN_A) = D_{\max}\left(\mN_A \big\| \Delta_A[\mN_A]\right) =\log t
\ee
 for some optimal $t$ that satisfies $t\Delta[\mN] \geq \mN$.
 Also, let $m = \ceil{t} $ so that $ m \Delta[\mN] \geq \mN$ still holds, and let $R_1$ be a static system with dimension $|R_1|=m$.
 We now construct the following supermap. For any state $\rho\in\md(R_1)$
 \begin{align}
 	\Omega_{R_1\to A}[\rho_{R_1}] \eqdef \frac{m}{m-1}\Big( \tr[\phi^+_{R_1} \rho_{R_1}] - \frac{1}{m} 											\Big) \mN_A \; + \; \frac{m}{m-1}\Big( 1 - \tr[\phi^+_{R_1} \rho_{R_1}]  \Big) \Delta_A[\mN_A] 								 
 \end{align}

The supermap $\Omega_{R_1\to A}$ has several properties.  First, it satisfies $\Delta_A\circ\Omega_{R_1\to A}=\Omega_{R_1\to A}\circ\mD_{R_1}$. Indeed, for any density matrix $\rho\in\md(R_1)$ we have
\ba
\Delta_A\circ\Omega_{R_1\to A}[\rho_{R_1}] &= \frac{m}{m-1}\Big( \tr[\phi^+_{R_1} \rho_{R_1}] - \frac{1}{m} 											\Big) \Delta_A[\mN_A] \; + \; \frac{m}{m-1}\Big( 1 - \tr[\phi^+_{R_1} \rho_{R_1}]  \Big) \Delta_A[\mN_A]\\
&=\Delta_A[\mN_A]\;,
\ea
and
\ba
\Omega_{R_1\to A}\big[\mD_{R_1}(\rho_{R_1})\big] &= \frac{m}{m-1}\Big( \tr[\phi^+_{R_1}\mD_{R_1}( \rho_{R_1})] - \frac{1}{m} 											\Big) \mN_A \; + \; \frac{m}{m-1}\Big( 1 - \tr[\phi^+_{R_1} \mD_{R_1}(\rho_{R_1})]  \Big) \Delta_A[\mN_A]\\
&= \frac{m}{m-1}\Big( \frac{1}{m} - \frac{1}{m}\Big) \mN_A \; + \; \frac{m}{m-1}\Big( 1 - \frac{1}{m}  \Big) \Delta_A[\mN_A]\\
&=\Delta_A[\mN_A]\;,
\ea
so that $\Delta_A\circ\Omega_{R_1\to A}=\Omega_{R_1\to A}\circ\mD_{R_1}$. Second, $\Omega_{R_1\to A}$ is a superchannel since the above map can be expressed as
\be
\Omega_{R_1\to A}[\rho_{R_1}]\eqdef\tr[\phi^+_{R_1}\rho_{R_1}]\mN_A+\frac{1}{m-1}\left(1-\tr[\phi^+_{R_1}\rho_{R_1}]\right)\left(m\Delta_A[\mN_A]-\mN_A\right)
\ee
and $m\Delta_A[\mN_A]-\mN_A\geq 0$. Hence, $\Omega\in\disc({R_1\to A})$. Finally, observe that $\Omega_{R_1\to A}[\phi^+_{R_1}]=\mN_A$.
Hence, the existence of such $\Omega$ implies that
\be
C^{0}_{\disc}(\mN_A)\leq\log m=\log\lceil t\rceil\leq\log t+1= LR_\Delta(\mN_{A})+1\;.
\ee
This completes the proof.
\end{proof}

\subsection{Coherence cost of a channel} \label{coherence_cost_of_channel}

For any $\mN\in\cptp(A_0\to A_1)$ the smoothed coherence cost  is defined as
\begin{align}
&C^{\epsilon}_{\mf}(\mN_{A})\eqdef\min_{\mN'\in B_\epsilon(\mN)}C^{0}_{\mf}\left(\mN'_{A}\right)
\end{align}
where
\be
B_\epsilon(\mN_A)=\Big\{\mN'\in\cptp(A_0\to A_1)\;:\;\frac{1}{2}\|\mN'_{A}-\mN_{A}\|_{\diamond}\leq \epsilon\Big\}\;.
\ee
The coherence cost of the channel $\mN_A$ is given by
\be
C_{\mf}(\mN_{A})=\lim_{\epsilon\to 0^+}\lim_{n\to\infty}\frac{1}{n}C^{\epsilon}_{\mf}(\mN_{A}^{\otimes n})
\ee

\subsubsection{The cost under MISC}

\begin{theorem}
For $\mf=\misc$
\be
C_{\mf}(\mN_{A})=LR_{\mc}^{\infty}(\mN_{A})\;.
\ee
\end{theorem}

\begin{proof}
First, note that from~\eqref{exact_cost_bw_lr} it follows that
\be
LR^{\epsilon}_\mc(\mN_A) \leq C^{\epsilon}_{\mf}(\mN_A)\leq  LR^{\epsilon}_\mc(\mN_A) +1
\ee
Hence, 
\be
\frac{1}{n}LR^{\epsilon}_\mc(\mN_A^{\otimes n}) \leq \frac{1}{n}C^{\epsilon}_{\mf}(\mN_A^{\otimes n})\leq  \frac{1}{n}LR^{\epsilon}_\mc(\mN_A^{\otimes n}) +\frac{1}{n}
\ee
and the limit $n\to\infty$ concludes the proof.
\end{proof}

\subsubsection{The cost under DISC}

\begin{theorem}
For $\mf=\disc$
\be
C_{\mf}(\mN_{A})=LR_{\Delta}^{\infty}(\mN_{A})\;.
\ee
\end{theorem}

\begin{proof}
First, note that from~\eqref{del} it follows that
\be
LR^{\epsilon}_\Delta(\mN_A) \leq C^{\epsilon}_{\mf}(\mN_A)\leq  LR^{\epsilon}_\Delta(\mN_A) +1
\ee
Hence, 
\be
\frac{1}{n}LR^{\epsilon}_\Delta(\mN_A^{\otimes n}) \leq \frac{1}{n}C^{\epsilon}_{\mf}(\mN_A^{\otimes n})\leq  \frac{1}{n}LR^{\epsilon}_\Delta(\mN_A^{\otimes n}) +\frac{1}{n}
\ee
and the limit $n\to\infty$ concludes the proof.
\end{proof}

The lack of AEP for channels motivates us to consider a more liberal method for smoothing.

\subsection{Liberal Coherence Cost of a Channel}\label{liberal_coherence_cost_of_a_channel}

We define the liberal one-shot $\epsilon$-approximate coherence-cost as
\be
C^{\epsilon}_{\mf}(\mN_{A})\eqdef\max_{\varphi\in\md(RA_0)}C^{\epsilon,\varphi}_{\mf}(\mN_{A})\;,
\ee
where
\be
C^{\epsilon,\varphi}_{\mf}(\mN_{A})\eqdef\min_{\mN'_A\in B_\epsilon^{\varphi}(\mN_A)}C^{0}_{\mf}(\mN'_{A})\;,
\ee
and
\be 
B_\epsilon^{\varphi}(\mN_A)\eqdef\Big\{\mN'\in\text{CP}(A_0\to A_1)\;:\;\|\mN'_{A}(\varphi_{RA_0})-\mN_{A}(\varphi_{RA_0})\|_{1}\leq \epsilon\Big\}.
\ee

The liberal coherence cost is defined as
\ba
C^{(\infty)}_{\mf}\left(\mN_{A}\right)&\eqdef\lim_{\epsilon\to 0^+}\lim_{n\to\infty}\max_{\varphi\in\md(RA)}\frac{1}{n}C^{\epsilon,\varphi^{\otimes n}}_{\mf}\left(\mN^{\otimes n}\right)\\
&=\lim_{\epsilon\to 0^+}\lim_{n\to\infty}\max_{\varphi\in\md(RA)}\min_{\mN'\in B_\epsilon^{\varphi^{\otimes n}}(\mN^{\otimes n})}\frac{1}{n}C^{0}_{\mf}(\mN'_{A^n\to B^n})
\ea
One can interpret the above cost in the following way. For any pure state $\varphi\in\md(RA_0)$ (with $|R|=|A_0|$ and $\varphi$ is full Schmidt rank) we define a $\varphi$-norm 
\be
\|\mE_{A}\|_{\varphi}\eqdef\left\|\mE_{A}(\varphi_{RA_0})\right\|_1
\ee
The the liberal cost can also be expressed as
\be
C^{(\infty)}_{\mf}\left(\mN_{A}\right)=\lim_{\epsilon\to 0^+}\lim_{n\to\infty}\max_{\varphi\in\md(RA_0)}\min_{\left\|\mN'-\mN^{\otimes n}\right\|_{\varphi^{\otimes n}}\leq\epsilon}\frac{1}{n}C^{0}_{\mf}(\mN'_{A^n\to B^n})
\ee
That is, we smooth with the $\varphi^{\otimes n}_{RA_0}$-norm and then maximizing over all such norms.

\begin{theorem}
For $\mf=\misc$
\be
C^{(\infty)}_{\mf}\left(\mN_{A}\right)=D^{(\infty)}_\mc\left(\mN_{A}\right)
\ee
\end{theorem}

\begin{proof}
From~\eqref{exact_cost_bw_lr} it follows that that for any fixed $\varphi\in\md(RA_0)$ we have
\be\label{philr}
 LR^{\epsilon,\varphi}_\mc(\mN_A) \leq C^{\epsilon,\varphi}_{\mf}(\mN_A)\leq  LR^{\epsilon,\varphi}_\mc(\mN_A) +1
\ee
From~\eqref{philr} it follows that $C^{(\infty)}_{\mf}\left(\mN_{A}\right)=LR^{(\infty)}_\mc\left(\mN_{A}\right)$ so that the theorem follows from the AEP relation~\eqref{aep}.
\end{proof}

\subsection{One shot distillable Coherence}\label{distillable_coherence}

We now consider the problem of distilling an arbitrary channel into pure-state coherence using MISC and DISC. 
Let $\Theta \in \mf(A\to B_1)$ where $\mf = \misc \text{ or } \disc$, such that for any input channel $\mE_A$, the output is a state preparation channel $\mF_B \in \cptp(B_0 \to B_1)$ where $B_0$ is a trivial system.
For $\epsilon>0$ and $n = |B_1|$, define
\begin{equation}
\text{DISTILL}_\mf^\epsilon(\mN_A)=\log \max\{n\;:\:\bra{\phi_{B_1}^+}\Theta\left[\mN_A\right]\ket{\phi_{B_1}^+}>1-\epsilon,\;\;\Theta\in \mf(A\to B_1)\},
\end{equation}
which represents the largest coherence attainable by MISC or DISC within $\epsilon$-error.
For all $\mN \in \cptp(A_0 \to A_1)$, we can write
\ba\label{inner_prod_ThetaN}
\left\langle \phi_{B_1}^+\left|\Theta\left[\mN\right]\right|\phi_{B_1}^+\right\rangle &= \left\langle \phi_{B_1}^+ \left| \left(\tr_{A}\left[\J^\Theta_{AB_1}\left(\left(J^{\mN}_A\right)^T\otimes I_{B_1}\right)\right]\right)\right|\phi^+_{B_1}\right\rangle \\
&= \tr\left[ \J^{\Theta}_{AB_1}\left(\left(J^{\mN}_A \right)^T\otimes \phi^+_{B_1} \right)  \right].
\ea
Note that the space of all operators that are invariant under any permutation in the classical basis, is a linear combination of maximally mixed state, $u_{A_1}$ and maximally coherent state, $\phi^+_{A_1}$.
Any operator is permutation invariant if 
\be
\Pi_x\; \sigma\; \Pi_x^{\dagger} = \sigma \;\; \forall \; \text{ permutation matrices }\Pi_x
\ee
The permutation-twirling operation can be expressed in the following way (see for example \cite{GMS2009})
\be
\mT(\cdot) = \frac{1}{m!}\sum_{x}\Pi_x (\cdot) \Pi_x^{\dagger} \;\; \forall \; \Pi_x
\ee 
where $m$ is the dimension of the input system.
Observe that the output of the above permutation-twirling operation on any state is permutation invariant and so can always be represented as a linear combination of $\phi^+_{A_1}$ and $u_{A_1}$.
Hence, we can express the second equality in \eqref{inner_prod_ThetaN} as
\ba
\tr\left[ \J^{\Theta}_{AB_1}\left(\left(J^{\mN}_A \right)^T\otimes \phi^+_{B_1} \right)  \right] &= \tr\left[ \J^{\Theta}_{AB_1} \left(\left(J^{\mN}_A \right)^T\otimes \mT\left(\phi^+_{B_1}\right) \right)  \right]\\
				&= \tr\left[ \left(\id_A \otimes \mT\left(\J^{\Theta}_{AB_1}\right)\right)\left(\left(J^{\mN}_A \right)^T\otimes \phi^+_{B_1} \right)  \right]\\
\ea
where the second equality follows from the fact that $\mT$ is self-adjoint in the Hilbert-Schmidt inner product.
Hence, without loss of generality we can express the Choi matrix $\J^{\Theta}_{AB_1} $ in following way
\be
\J^{\Theta}_{AB_1}= \alpha_A \otimes \phi_{B_1}^+ + \frac{1}{n-1}\beta_A\otimes(I_{B_1}-\phi^+_{B_1})
\ee
where $n = |B_1|$ and $\alpha_A, \beta_A \in \herm(A)$ such that
 $\J^\Theta_{AB_1}\geq 0$, $\J_{A_1}^{\Theta}=I_{A_1}$, and $\J_{A}^{\Theta}=\J_{A_0}^{\Theta}\otimes u_{A_1}$.
In terms of $\alpha_A$ and $\beta_A$, we can write these conditions as
\begin{align}
\alpha_A,\beta_a &\geq 0, \\
\tr(\alpha_A + \beta_A)&=|A_1|\\
\alpha_A + \beta_A &=\tr_{A_1}(\alpha_A + \beta_A)\otimes u_{A_1}.
\end{align}
From the MISC condition of $\mD_{AB}(\J^\Theta_{AB})=\mD_A\otimes\id_B(\J^\Theta_{AB})$, we get
\begin{equation}
\mD(\alpha_A)(n-1)=\mD(\beta_A).
\end{equation}
Defining $\beta_A = \rho_{A_0}\otimes u_{A_1} - \alpha_A$ where $\rho_{A_0}=\frac{1}{|A_1|}\tr_{A_1}(\alpha_A + \beta_A)$.
Since $\tr[\rho_{A_0}] = 1$, $\rho_{A_0}$ is a density matrix. 
So, we can rewrite these constraints as
\begin{align}
\alpha_A&\geq 0\label{Eq:MISC-distill-cons0}\\
\rho_{A_0}\otimes I_{A_1}&\geq \alpha_A,\label{Eq:MISC-distill-cons1}\\
\frac{1}{n}\mD(\rho_{A_0})\otimes I_{A_1}&=\mD(\alpha_A),\label{Eq:MISC-distill-cons2}\\
\rho_{A_0}&\in \md(A_0)\label{Eq:MISC-distill-cons3}
\end{align}
We can also consider imposing the additional DISC constraint of  $\id_A\otimes\mD_{B}(\J^\Theta_{AB})=\mD_A\otimes\id_B(\J^\Theta_{AB})$ which gives
\begin{equation}
\alpha_A + \beta_A=\mD(\alpha_A + \beta_A).
\end{equation}
This amounts to replacing Eq. \eqref{Eq:MISC-distill-cons1} with the condition
\begin{align}
\label{Eq:DISC-distill-cons1}
n\mD(\alpha_A)\geq \alpha_A.
\end{align}
Next notice that we can always write $\alpha_A = \mD_A(\alpha_A) + \gamma_A$ for some $\gamma_A$ with zeroes on the diagonal.
Then, since $\tr_{A_1}\left[ \left(J^{\mN}_A \right)^T\right] = I_{A_0}$, we can write
\ba
\tr\left[\alpha_A \left( J^{\mN}_A \right)^T  \right] &= \tr\left[\left(\mD_A\left(\alpha_A\right) + \gamma_A\right) \left( J^{\mN}_A \right)^T  \right]\\
				&= \tr\left[\mD_A\left(\alpha_A\right) \left( J^{\mN}_A \right)^T  \right] + \tr\left[\gamma_A \left( J^{\mN}_A \right)^T  \right]\\
				&= \tr\left[\left( \frac{1}{n}\mD(\rho_{A_0})\otimes I_{A_1}\right) \left( J^{\mN}_A \right)^T  \right] + \tr\left[\gamma_A \left( J^{\mN}_A \right)^T  \right]\\
				&= \frac{1}{n} + \tr\left[\gamma_A \left( J^{\mN}_A \right)^T  \right]\\
\ea

Hence, we have the following one-shot distillable rates.
\begin{theorem}
For $\mf = \misc \text{ or } \disc$
\be
\text{\rm DISTILL}_\mf^\epsilon(\mN)=\log\,\max\; n
\ee
such that 
\ba 
\tr\left[\gamma_A\left(J^{\mN}_A \right)^T \right] &\geq 1 - \frac{1}{n} - \epsilon\; ,\\
\mD_A(\gamma_A) &= 0\; ,\\
\rho_{A_0} &\in \md(A_0)\; ,\\
\left[\rho_{A_0} - \frac{1}{n}\mD_{A_0}\left(\rho_{A_0} \right) \right]\otimes I_{A_1} \geq \gamma_A &\geq -\frac{1}{n}\mD_{A_0}\left(\rho_{A_0} \right)\otimes I_{A_1} \;\; \text{(specifically for }\mf = \misc \text{)} \; ,\\
\frac{n-1}{n}\mD_{A_0}(\rho_{A_0})\otimes I_{A_1} \geq \gamma_{A} &\geq -\frac{1}{n}\mD_{A_0}\left( \rho_{A_0}\right)\otimes I_{A_1} \; \; \text{(specifically for }\mf = \disc\text{)}.
\ea
\end{theorem}
\begin{remark} Note that $D_\misc^\epsilon(\mN)=D_\disc^\epsilon(\mN)$ when $|A_0|=1$, and their common rate matches that given in Refs. \cite{Regula-2018a, Zhao-2019a} for distilling coherence from static resources (i.e. states).  However for channels, the MISC and DISC distillable coherence can possibly differ.  We leave it as an open problem to find channels that have such a property.
\end{remark}

\medskip
\begin{example}
Let us consider the partially depolarizing channel $\mN^{\text{dep}}_{\lambda,d}:\mB(A_1)\to\mB(A_1)$,
\begin{equation}
\mN^{\text{dep}}_{\lambda,d}(\chi)=\lambda \chi+(1-\lambda)\tr[\chi]u_{A_1}. 
\end{equation}
where $d = |A_1|$.
The Choi matrix of this channel is given by
\begin{equation}
J^{\mN^{\text{dep}}}_{A_1 \tA_1}=\lambda \phi^+_{A_1 \tA_1} +\frac{1-\lambda}{d}I_{A_1 \tA_1},
\end{equation} 
 We exploit the symmetry by noting that both $\phi_{A_1 \tA_1}^+$ and $I_{A_1 \tA_1}$ are $U^*\otimes U$ invariant. 
 We restrict our twirling to an average over the group of incoherent unitaries, i.e., each $U$ involves a permutation and/or a change in relative phase.
 Note that dephasing commutes with this operation so if Eq. \eqref{Eq:MISC-distill-cons2} holds before the twirl, it will also hold after.
 The action of twirling will convert $\rho_{A_1}\otimes I_{\tA_1}\to u_{A_1}\otimes I_{\tA_1}$ while converting $\alpha_A$ into an operator of the form
\begin{align}
\alpha_{A_1\tA_1} &= p\sum_{i\not=j}\op{ij}{ij}+q\sum_{i}\op{ii}{ii}+ r \sum_{i\not= j}\op{ii}{jj}\notag\\
&=p\sum_{i\not=j}\op{ij}{ij}+(q - r)\sum_{i}\op{ii}{ii}+ r \phi^+_{A_1\tA_1}.
\end{align}
The eigenvalues of $\alpha_{A_1\tA_1}$ are easily seen to be $\{p,q - r, q-r+rd\}$, and so equations \eqref{Eq:MISC-distill-cons0} and \eqref{Eq:MISC-distill-cons1} require that $p,q-r\geq 0$ and $p, q - r+rd\leq\frac{1}{d}$.
 From equation \eqref{Eq:MISC-distill-cons2}, we must also have $p=q=\frac{1}{nd}$.  With these constraints in place, our goal is to maximize $n$ such that
\begin{align}
\tr\left[\alpha^T_{A_1\tA_1} J^{\mN^\text{dep}}_{A_1\tA_1}\right]=(1-\lambda)\frac{(d-1)}{nd}+\left(\frac{1}{nd}-r\right)(\lambda d+(1-\lambda))+r(\lambda d^2+(1-\lambda)).
\end{align}
This function is strictly increasing w.r.t. $r$, and the constraints necessitate that $r\leq\min\{\frac{n-1}{d-1}\frac{1}{nd},\frac{1}{nd}\}$.  So when $n\leq d$, we take $r=\frac{n-1}{d-1}\frac{1}{nd}$ and obtain
\begin{align}
\label{Eq:dep-distill-1}
\tr\left[\alpha^T_{A_1\tA_1}J^{\mN^\text{dep}}_{A_1\tA_1}\right]&=(1-\lambda)\frac{d-1}{nd}+\frac{d-n}{nd(d-1)}(\lambda d+(1-\lambda))+\frac{n-1}{nd(d-1)}(\lambda d^2+(1-\lambda))\notag\\
&=(1-\lambda)\frac{1}{n}+\lambda.
\end{align}
Notice that when $\lambda=1$ we obtain $\tr\left[\alpha^T_{A_1\tA_1}J^{\mN^\text{dep}}_{A_1\tA_1}\right]=1$.  This says that $\log n$ bits can be perfectly distilled, which is expected: the free superchannel just consists of inputting $\phi^+_{A_1\tA_1}$ into the given channel and then as post-processing performs a MIO map that converts $\phi^+_{A_1 \tA_1}$ into $\phi^+_{B_1}$.  On the other hand, if $n\geq d$, we take $r =\frac{1}{nd}$ and Eq. \eqref{Eq:dep-distill-1} becomes
\begin{equation}
\tr\left[\alpha^T_{A_1\tA_1}J^{\mN^\text{dep}}_{A_1\tA_1}\right]=(1-\lambda)\frac{1}{n}+\frac{d}{n}\lambda.
\end{equation}
Notice also that in this case our optimizer $\rho_{A_0}$ is completely dephased, which means our solution for MISC is also the solution for DISC.  We summarize our findings as follows.
\begin{lemma}
For the partial depolarizing channel $\mN^{\text{dep}}_{\lambda,d}$ and $0\leq\epsilon<1$,
\begin{equation}
\text{\rm DISTILL}_\misc^\epsilon\left(\mN^{\text{dep}}_{\lambda,d}\right)=\text{\rm DISTILL}_\disc^\epsilon\left(\mN^{\text{dep}}_{\lambda,d}\right)=\begin{cases}\log\lfloor\frac{1-\lambda}{1-\lambda-\epsilon}\rfloor\quad\text{if $\epsilon<\frac{(d-1)(1-\lambda)}{d}$}\\
\log\lfloor\frac{1-\lambda+\lambda d}{1-\epsilon}\rfloor\quad\text{if $\epsilon\geq \frac{(d-1)(1-\lambda)}{d}$}
\end{cases}.
\end{equation}
\end{lemma}

\end{example}

\begin{example}
We next consider the partial dephasing channel $\mN^{\Delta}_{\lambda,d}:\mB(A_1)\to\mB(A_1)$,
\begin{equation}
\mN^{\Delta}_{\lambda,d}(\chi)=\lambda \chi + (1-\lambda)\mD(\chi).
\end{equation}
The Choi matrix of this channel is given by
\begin{equation}
J^{\mN^{\text{dep}}}_{A_1\tA_1}=\lambda \phi^+_{A_1 \tA_1} +(1-\lambda)\sum_{i=1}^d\op{ii}{ii}.
\end{equation} 
By the same argument as before, we can assume without loss of generality that $\alpha_A$ has the form
\begin{align}
\alpha_{A_1\tA_1}&=p\sum_{i\not=j}\op{ij}{ij}+(q-r)\sum_{i}\op{ii}{ii}+r \phi^+_{A_1\tA_1}.
\end{align}
However this time the fidelity with $\phi^+_{B_1}$ is given by
\begin{align}
\tr\left[\alpha^T_{A_1\tA_1}J^{\mN^\Delta}_{A_1\tA_1}\right]=\left(\frac{1}{nd}-r\right)d+r(\lambda d^2+(1-\lambda)d).
\end{align}
Again, the constraints of the problem demand $r\leq\min\{\frac{n-1}{d-1}\frac{1}{nd},\frac{1}{nd}\}$.  When $n\leq d$, it holds that% $\tr\left[\alpha^TJ^{\mN^\Delta}_{A_1\tA_1}\right]$ is maximum at
\begin{align}
\tr\left[\alpha^T_{A_1\tA_1}J^{\mN^\Delta}_{A_1\tA_1}\right]&=\frac{d-n}{n(d-1)}+\frac{n-1}{n(d-1)}(\lambda d+(1-\lambda))\notag\\
&=\frac{1+(n-1)\lambda}{n}.
\end{align}
On the other hand, when $n\geq d$, we take $r=\frac{1}{nd}$ to obtain
\begin{align}
\tr\left[\alpha^T_{A_1\tA_1}J^{\mN^\Delta}_{A_1\tA_1}\right]&=\frac{\lambda d+(1-\lambda)}{n}.
\end{align}
These are the same maximum fidelities as the depolarizing channel, and we therefore have the following conclusion.
\begin{lemma}
For the partial dephasing channel $\mN^{\Delta}_{\lambda,d}$ and $0\leq\epsilon<1$,
\begin{equation}
\text{\rm DISTILL}_\misc^\epsilon\left(\mN^{\Delta}_{\lambda,d}\right)=\text{\rm DISTILL}_\disc^\epsilon\left(\mN^{\Delta}_{\lambda,d}\right)=\begin{cases}\log\lfloor\frac{1-\lambda}{1-\lambda-\epsilon}\rfloor\quad\text{if $\epsilon<\frac{(d-1)(1-\lambda)}{d}$}\\
\log\lfloor\frac{1-\lambda+\lambda d}{1-\epsilon}\rfloor\quad\text{if $\epsilon\geq \frac{(d-1)(1-\lambda)}{d}$}
\end{cases}.
\end{equation}
\end{lemma}

\end{example}

\section{Outlook and Conclusions}\label{conclusions}

In this paper, we have developed the resource theory of dynamical coherence using the classical channels as free channels.
In previous works on the quantum resource theory of dynamical coherence \cite{DKW+2018, BGM+2017, CH2016, LY2019, TEZ+2019}, the free channels were taken to be the free operations from the QRT of static coherence, like MIO, IO, etc.  However, it is known that most of these operations have the ability to distribute coherence over time and space.  As this is quite powerful for quantum information processing, they do not seem suitable for formulating a dynamical resource theory of coherence.  In contrast, we argue that 
%In the resource theory of static coherence, the free states are the classical states and the operations that leave this set invariant form the free operations.  It is known that most of the free operations in the static case require coherence or coherence-generating processes to be implemented. Hence, when we generalize the resource theory of coherence from the state case to the dynamical case, and we use the free channels from the former as the free channels in the latter, then the free channels in the dynamical case themselves need coherence to be implemented and thus, cannot be used to quantify coherence of a channel.  Therefore, 
a proper extension of the QRT of coherence should require the free channels be void of any coherence-preserving power.
So, the classical channels come as a natural choice and for the first time, we overcome the problem of using coherence in free channels.
Note that the T-gate (in quantum computation) is not free and even the quantum identity channel is not free as the preservation of coherence should be considered a resource.

Similar to the static QRT of coherence where the free operations can have a non-free dilation, in our work on dynamical QRT of coherence, the free superchannels can have a non-free realization.
That means, the pre- and post-processing channels need not be classical.
The only requirement on the set of free superchannels comes from the golden rule of QRT.
This implies that the free superchannels must never generate coherent channels when the input channels are classical even when tensored with identity, i.e., even when the free superchannel acts on a part of the input classical channel.
This enlargement of the set of free superchannels is necessary for a meaningful resource theory of coherence.
Take for example the set of free superchannels which can be realized only by classical pre- and post-processing channels.
In this case, the output channel is always classical irrespective of the input channel, eliminating all the advantage offered by a quantum channel.
Thus, this very small set can not be used to study the resource theory of quantum coherence.

In section \ref{free_superchannels}, we start by defining four sets of free superchannels. 
We name them as maximally incoherent superchannels (MISC), dephasing-covariant incoherent superchannels (DISC), incoherent superchannels (ISC), and strictly incoherent superchannels (SISC).
 We show that the set of free superchannels in the dynamical resource theory of coherence can be characterized analogous to the free channels in the static resource theory of coherence. 
We also show that MISC and DISC can be characterized just on the basis of their Choi matrices and dephasing channels which is given in Eq. \eqref{thmmisc} and \eqref{disc_condition} for MISC and DISC, respectively. 

Section \ref{quantification} then deals with the quantification of dynamical coherence.
In section \ref{complete_family_of_monotones}, we find the complete set of monotones for MISC and DISC.
That means, to see if we can convert from one quantum channel to another, it is sufficient to check if all the monotones of this (complete) set acting on one channel are greater than the other.
A complete family of monotones for a general resource theory of processes was presented in \cite{GM2019}.
It is, in general, a hard problem to compute these functions and in some cases like LOCC-based entanglement, it is even NP-hard.
We show that for the resource theory of dynamical coherence, these functions (under MISC and DISC) can be computed using an SDP (Eq. \eqref{complete_monotones_SDP}).
Next, in section \ref{relative_entropies}, we also find monotones that are bases on relative entropy.
In \cite{GW2019}, Gour and Winter showed that the generalization of relative entropy from states to channels is not unique.
In their work, they listed six relative entropies as measures of dynamical resources.
They also introduced a new type of smoothing called ``liberal" smoothing.
We show in section \ref{relative_entropies} that out of these relative entropies defined in \cite{GW2019}, three relative entropies clearly form a monotone under MISC and DISC.
We then discuss about various log-robustness of coherence of channels which are based on the max-relative entropy of channels, $D_{\max}$ and show that it can be computed with an SDP (Eq. \eqref{lr_as_sdp}).
For the qubit case, we calculated the log-robustness of coherence for classical channels, identity channel, replacement channel,depolarizing channels, and unitary channels. 
We also show that the log-robustness of coherence of channels is additive under tensor product (Lemma \ref{additivity_of_lr_LEMMA}).
We then define a ``liberally'' smoothed log-robustness of coherence which when regualarized is equal to a regularized relative entropy introduced in \cite{GW2019} (i.e., it satisfies AEP), and behaves monotonically under completely resource non-generating superchannels. .

The next section is dedicated to the problem of interconversion of one resource into another.
In section \ref{interconversion_as_SDP}, we define a conversion distance between two channels (Eq. \eqref{d_N_to_M}).
A given channel can be simulated using another if the interconversion distance is very small.
For MISC and DISC, we showed that the interconversion distance can be computed using an SDP (Theorem \ref{misc_disc_sdp}).
We then calculated the exact, asymptotic, and liberal cost of coherence of a channel and found that the liberal cost of coherence is equal to a variant of regularized relative entropy. Lastly, in this section, we also define the one-shot distillable coherence for MISC and DISC, and calculate it for partial depolarizing and partial dephasing channels. 

Due to the realization of a superchannel as a pre- and post-processing channel, there are added complexities in the generalization of a quantum resource theory of states to channels as mentioned in \cite{GM2019}.
In our case, we see that the simple generalizations don't work.
For example, while calculating coherence costs, we had to introduce the concept of liberal cost (based on liberal smoothing as defined in \cite{GW2019}) to show it to be equal to a relative entropy.

Clearly, our work is just a start of a whole unexplored field of the quantum resource theory of dynamical coherence.
For instance, one can solve for interconversion, cost etc. for ISC and SISC.
One can define more sets of superchannels analogous to how various free operations are defined in the static case.
We also leave as open the problem of finding an example of a channel where the MISC and DISC distillable coherence are different.
In section \ref{distillable_coherence}, we worked out the distillable coherence for the partial depolarizing channel and the partial dephasing channel and found no difference for MISC and DISC case.

\textit{Note added.} Recently, we became aware of the work \cite{H2019} which considers resource preserving channels as a resource in a general resource theory.

\begin{acknowledgments}
The authors would like to thank Karol Horodecki, Varun Narasimhachar, Yunlong Xiao, and Carlo Maria Scandolo for helpful discussions.
G.G. and G.S. acknowledge support from the Natural Sciences and Engineering Research Council of Canada (NSERC).  EC is supported by the National Science Foundation (NSF) Award No. 1914440.
\end{acknowledgments}

\bibliography{QRTbib}

\begin{appendix}

\section{Proof of dual of the log-robustness}

Finding the dual of the log-robustness($LR_{\mc}(\mN_A)$) is equivalent to finding the dual of $2^{LR_{\mathfrak{C}}(\mathcal{N}_A)}$.
From (\ref{lr_as_sdp}), we can write $2^{LR_{\mathfrak{C}}(\mathcal{N}_A)}$ as
\begin{equation}
   \min\Big\{\frac{1}{|A_0|}\tr[\omega_A]\;:\;\omega_A\geq J^\mN_A\;\;\,\;\;\mD_A[\omega_A]=\omega_A\;\;,\;\;\omega_{A_0}=\tr[\omega_A]u_{A_0}\;\;,\;\;\omega_A\geq 0\Big\}
\end{equation}
where $u_{A_0} = \frac{I_{A_0}}{|A_0|}$.
The primal problem of the above conic linear program can be stated as

\begin{equation}
	 \min \Big\{ \frac{1}{|A_0|} \tr[\omega_{A} I_A] \; : \;
    \Gamma(\omega_A) - H_2 \in \mathfrak{K}_2 \; , \;
    \omega \geq 0
    \Big\}
\end{equation}
         
where $\Gamma(\omega_A)$ is a linear map and is expressed as a 3-tuple such that $\Gamma(\omega_A) = (\omega_{A_0} - \tr[\omega_{A_0}]u_{A_0} \; , \; \omega_{A} \; , \; \omega_{A} - \mathcal{D}(\omega_{A}))$. The separation of elements in the tuple can be understood as a direct sum between the subspaces in a larger vector space.
Likewise, $H_2$ is also expressed as a 3-tuple such that $H_2 = (0_{A_0} , J^{\mathcal{N}}_A , 0_A )$.
The cone $\mathfrak{K}_2$ can be expressed as a 3-tuple as $\mathfrak{K}_2 = \{(0_{A_0}\, , \, \zeta_{A}\, , \, 0_{A}) : \zeta_A \geq 0   \}$.
Hence, the dual cone $\mathfrak{K}_2^{*} = \{(Z_{A_0}, \beta_{A}, W_{A}) :  Z_{A_0}\in \herm(A_0),\; \beta_A \geq 0,\; W_A \in \herm(A)    \}$.

Therefore, it is easy to see that the dual to the above primal problem is
\begin{equation}
		\max \Big\{  \frac{1}{|A_0|} \tr[\beta_A J^{\mathcal{N}}_A] \; : \;
         		 I_A - \Gamma^{*}(Z_{A_0},\beta_A, W_A) \geq 0 \; , \;
                 Z_{A_0} \in \herm(A_0)\; , \; W_A \in \herm(A) \; , \;
                 \beta_A \geq 0 
             \Big\}
\end{equation}
In order to find $\Gamma^*(Z_{A_0},\beta_A, W_A)$, we need to equate 
\begin{equation}\label{dual_map_Gamma}
    \tr[(Z_{A_0}, \beta_A, W_A)\Gamma(\omega_A)] = \tr[\Gamma^* (Z_{A_0}, \beta_A, W_A) \omega_A]
\end{equation}
From the LHS of (\ref{dual_map_Gamma}), we find
\begin{equation}
     \tr[(Z_{A_0}, \beta_A, W_A)\Gamma(\omega_A)] = \tr[Z_{A_0} (\omega_{A_0} - \tr[\omega_{A_0}]u_{A_0})] + \tr[\beta_A \, \omega_{A}]+ \tr[W_A \, (\omega_{A} - \mathcal{D}(\omega_{A}))]
\end{equation}
Therefore,
\begin{equation}
     \Gamma^{*}(Z_{A_0},\beta_A, W_A) =
     Z_{A_0} \otimes I_{A_1} - \tr[Z_{A_0}]u_{A_0} \otimes I_{A_1} +  \beta_{A} + W_{A} - \mathcal{D}(W_{A})
\end{equation}
So, we can rewrite the first constraint in the dual problem as
\begin{equation}\label{first_dual_condition}
    I_{A} - Z_{A_0} \otimes I_{A_1} + \tr[Z_{A_0}]u_{A_0} \otimes I_{A_1} -  \beta_{A} - W_{A} + \mathcal{D}(W_{A}) \geq 0
\end{equation}
%Rewriting $ I_{A} - Z_{A_0} \otimes I_{A_1} + \tr[Z_{A_0}]u_{A_0} \otimes I_{A_1} - W_{A} + \mathcal{D}(W_{A})$ as $\eta_A$ where $\eta_A$ satisfies
Now let $\eta_A \geq 0$ obey the following conditions
\be
\mD_{A}(\eta_{A})=\mD_{A_0}\left(\eta_{A_0}\right)\otimes u_{A_1}\;\; , \;\;\mD_{A_1}[\eta_{A_1}]=I_{A_1}\;
\ee
Any such matrix can be expressed as $\left(I_{A_0} - Z_{A_0} + \tr[Z_{A_0}]u_{A_0}\right) \otimes I_{A_1} - W_{A} + \mathcal{D}(W_{A})$.
Hence, we can express (\ref{first_dual_condition}) as

\begin{equation}
    \eta_A \geq \beta_{A} \geq 0
\end{equation}

Since, $J^{\mathcal{N}}_A \geq 0$, therefore from the above equation we get
\begin{equation}
    \tr[\, \eta_A \, J^{\mathcal{N}}_A \, ] \, \geq \, \tr[\, \beta_A \, J^{\mathcal{N}}_A \, ]
\end{equation}

Hence, we can recast the dual problem in the following form 
\begin{equation}\label{dual_D_max}
          \max \Big\{ \frac{1}{|A_0|} \tr[\, \eta_A \, J^{\mathcal{N}}_A \, ] \; : \;
   				\mD_A(\eta_A) = \mD_{A_0}(\eta_{A_0}) \otimes u_{A_1} \; , \; \mD_{A_1}[\eta_{A_1}] = I_{A_1} \; , \; \eta_A \geq 0
   				\Big\}
\end{equation}

Therefore,
\be
	LR_{\mc}(\mN_A) = \log \max \Big\{ \frac{1}{|A_0|} \tr[\, \eta_A \, J^{\mathcal{N}}_A \, ] \; : \;
   				\mD_A(\eta_A) = \mD_{A_0}(\eta_{A_0}) \otimes u_{A_1} \; , \; \mD_{A_1}[\eta_{A_1}] = I_{A_1} \; , \; \eta_A \geq 0
   				\Big\}
\ee

which is Eq.(\ref{dual_lr}).

\section{Proof of Theorem \ref{misc_disc_sdp} and the dual of the conversion distance for MISC and DISC}

In \cite{W2009}, it was shown that the diamond norm can be expressed as the following SDP
\be \label{diamond_norm}
 \frac{1}{2} \|\mE_B - \mF_B \|_{\diamond} = \min_{\omega \geq 0 ; \omega \geq J^{\mE - \mF}_B} \|\omega_{B_0}\|_{\infty} \; \forall \; \mE ,\mF  \in \cptp(B_0 \to B_1)\; .
\ee
 Note that \eqref{diamond_norm} can be rewritten as \cite{GW2019}
\be
 \frac{1}{2} \|\mE_B - \mF_B \|_{\diamond} = \min\{ \lambda \; : \; \lambda\mQ_B \geq \mE_B - \mF_B \; ; \mQ_B \in \cptp(B_0 \to B_1)\}\; .
\ee
Taking $\mE_B = \Theta_{A \to B}[\mN_A]$  and $\mF_B = \mM_B$, $d_{\mf}(\mN_A \to \mM_B)$ in \eqref{d_N_to_M} becomes 
\be \label{d_sdp_gilad}
  d_{\mf}(\mN_A \to \mM_B) = \min\{ \lambda \; : \; \lambda\mQ_B \geq \Theta_{A \to B}[\mN_A] - \mM_B \; , \mQ_B \in \cptp(B_0 \to B_1) \; , \; \Theta \in \mf(A \to B) \} \; .
\ee

For the case $\mf = \misc$, let us start by denoting $\omega_B$ as the Choi matrix of $\lambda \mQ_B$ and $\alpha_{AB}$ as the Choi matrix of $\Theta$, we can express $d_{\mf}(\mN_A \to \mM_B)$ as
\ba
 &d_{\mf}\left(\mN_{A}\to\mM_B\right) = \min  \lambda \\
 &\text{subject to :} \; (1) \; \lambda I_{B_0} \geq \omega_{B_0}\; , \; 
 								 (2) \; \omega_B \geq 0 \; , \;
 								 (3) \; \omega_B \geq \tr_A\left[\alpha_{AB}\left( (J^{\mN}_A)^T \otimes I_B \right)\right] - J^{\mM}_B \enspace , \\
	& \qquad \qquad \quad \; (4) \;  \alpha_{AB} \geq 0 \; , \;
								 (5) \; \alpha_{AB_0} = \alpha_{A_0 B_0} \otimes u_{A_1} \; , \;
								 (6) \; \alpha_{A_1B_0} = I_{A_1B_0}\; ,\\
 	& \qquad \qquad \quad \; (7)\; \tr[\alpha_{AB}X^{i}_{AB}] = 0 \; \forall \; i =  1, \ldots, n \; 
\ea
where $n \equiv |AB|(|B| - 1)$ and $\{X^i_{AB}\}_{i=1}^n$ are the bases of the subspace $\mk_{\mf}$ defined in~\eqref{k_MISC}.
Here, constraints (1-3) are due to diamond norm, constraints (4-6) follow from the requirement of $\Theta$ to be a superchannel and constraint (7) is due to the requirement that $\Theta \in \mf$.

Now consider a linear map
$\mL \, : \, \mathbb{R}\,\oplus \, \herm(B)\, \oplus\, \herm(AB) \to \herm(B_0)\,\oplus\,\herm(B)\,\oplus\,\herm(AB_0)\,\oplus\,\herm(A_1 B_0)\, \oplus^n\, \mathbb{R}   $ where $\oplus^n\, \mathbb{R}$ denotes $\underbrace{\mathbb{R} \oplus \ldots \oplus \mathbb{R}}_{n}$.

Its action on a generic element $\mu = (\lambda \, ,\, \omega_B \, ,\, \alpha_{AB}) $ of $\mathbb{R}\oplus \herm(B) \oplus \herm(AB)$ such that $\lambda \in \mathbb{R}_+\; , \; \omega_B \geq 0\; , \; \alpha_{AB} \geq 0  $ is
\be
  \mL(\mu) \eqdef \Big(\lambda I_{B_0} - \omega_{B_0} \; , \; \omega_B - \tr\left[\alpha_{AB} \left( (J^{\mN}_A)^T \otimes I_B \right) \right]  \; , \; \alpha_{AB_0} - \alpha_{A_0 B_0} \otimes u_{A_1} \; , \; \alpha_{A_1 B_0} \; , \; \tr[\alpha_{AB}X^{1}_{AB}]\; , \;\ldots \; , \tr[\alpha_{AB}X^{n}_{AB}]   \Big)
\ee
Taking a generic element $\nu = (\beta_{B_0} \; , \; \gamma_B \; , \; \tau_{AB_0}\; ,\; \zeta_{A_1 B_0}\; , t_1 , \ldots  , t_n)$ of $\herm(B_0)\oplus\herm(B)\oplus\herm(AB_0)\oplus\herm(A_1 B_0) \oplus^n \mathbb{R}$ such that $\beta_{B_0} \geq 0$\, ,  $\gamma_B \geq 0$\, , we have
\be
\mL^*(\nu) = \Big( \tr[\beta_{B_0}]\; , \; \gamma_{B} - \beta_{B_0}\otimes I_{B_1} \; , \; \tau_{AB_0} \otimes I_{B_1} - (J^{\mN}_A)^T \otimes \gamma_{B}  - \tau_{A_0 B_0}\otimes u_{A_1}\otimes I_{B_1}   + \tau_{A_1 B_0} \otimes I_{A_0 B_1} + \sum_i t_i X^{i}_{AB} \Big) \; .
\ee
Following \citep{B2002}, the dual is given by
\be
 d_{\mf}\left(\mN_{A}\to\mM_B\right) = \max \big\{ -\tr\left[J^{\mM}_B \gamma_{B}\right] + \tr\left[ \zeta_{A_1 B_0}   \right] \big\}
\ee
 where the maximum is subject to 
\ba
&\beta_{B_0}\otimes I_{B_1} \geq \gamma_{B} \geq 0 \; , \; 1 \geq \tr[\beta_{B_0}] \; ,\\
&\zeta_{A_1 B_0} \in \herm(A_1 B_0) \; , \; \tau_{AB_0} \in \herm(AB_0)  \; , \; t_1 ,\ldots,t_n \in \mathbb{R}\; , \\
&J^{\mN}_A\otimes \gamma_B + \tau_{A_0 B_0}\otimes u_{A_1} \otimes I_{B_1} - \tau_{AB_0}\otimes I_{B_1} - \tau_{A_1 B_0}\otimes I_{A_0 B_1} - \sum_i t_i X^{i}_{AB} \geq 0\; .
\ea

For the case of $\mf = \disc$, note that the only  distinction is in the choice of basis of the subspace $\mk_{\mf}$.
So, in this case, the dual is given by
\be
 d_{\mf}\left(\mN_{A}\to\mM_B\right) = \max \big\{ -\tr\left[J^{\mM}_B \gamma_{B}\right] + \tr\left[ \zeta_{A_1 B_0}   \right] \big\}
\ee
 where the maximum is subject to 
\ba
&\beta_{B_0}\otimes I_{B_1} \geq \gamma_{B} \geq 0 \; , \; 1 \geq \tr[\beta_{B_0}] \; ,\\
&\zeta_{A_1 B_0} \in \herm(A_1 B_0) \; , \; \tau_{AB_0} \in \herm(AB_0)  \; , \; t_1 ,\ldots,t_n \in \mathbb{R}\; , \\
&J^{\mN}_A\otimes \gamma_B + \tau_{A_0 B_0}\otimes u_{A_1} \otimes I_{B_1} - \tau_{AB_0}\otimes I_{B_1} - \tau_{A_1 B_0}\otimes I_{A_0 B_1} - \sum_i t_i Y^{i}_{AB} \geq 0\; .
\ea
Therefore, we see that $d_{\mf}(\mN_A \to \mM_B)$ is an SDP in the dynamical resource theory of quantum coherence if the free superchannels belong to MISC or DISC.

\end{appendix}

\end{document}